\documentclass{article}
\usepackage[a4paper]{geometry}
\usepackage{authblk}
\usepackage{amsmath}
\usepackage{amssymb}
\usepackage{latexsym}
\usepackage{amsthm} 
\usepackage{array}
\usepackage{stmaryrd}
\usepackage{xspace} 
\usepackage{paralist}
\usepackage{algorithm, algorithmic}
\usepackage{graphicx} 
\usepackage{subfig}
\usepackage{multirow} 
\usepackage{tikz}
\usetikzlibrary{shapes,arrows,positioning,chains,decorations.pathreplacing,calc,patterns,automata}
\tikzset{file/.style={regular polygon, regular polygon sides=6, draw, very thick, fill=gray!20,inner sep=6pt}}
\tikzset{printer/.style={trapezium, trapezium left angle=90, shape border rotate=90, draw, very thick, fill=gray!20}}
\tikzset{group/.style={trapezium, draw, very thick, fill=gray!20, shape border rotate=180}}
\tikzset{folder/.style={rectangle, draw, very thick, fill=gray!20}}
\tikzset{project/.style={circle, draw, very thick, fill=gray!20,inner sep=4pt}}
\tikzset{user/.style={trapezium, draw, very thick, fill=gray!20}}
\tikzset{elbl/.style={inner sep=2pt,fill=red!10,font=\footnotesize,font=\sffamily}}
\tikzset{loop on left/.style={to path={.. controls +(-1,1) and +(-1,-1) .. (\tikztotarget) \tikztonodes}}}
\tikzset{loop on right/.style={to path={.. controls +(1,1) and +(1,-1) .. (\tikztotarget) \tikztonodes}}}
\tikzset{loop below/.style={to path={.. controls +(-1,-1) and +(1,-1) .. (\tikztotarget) \tikztonodes}}}
\newcommand{\interestahead}{*}
\newcommand{\interestdhead}{square}

\newcommand{\semantics}[1]{\llbracket #1 \rrbracket}

\newtheorem{definition}{Definition}
\newtheorem{example}{Example}

\newtheorem{remark}{Remark}

\newtheorem{corollary}{Corollary}
\newtheorem{proposition}{Proposition}

\newcommand{\set}[1]{\ensuremath{\left\{#1\right\}}} 
\newcommand{\card}[1]{\left| #1 \right|}
\newcommand{\RcoR}{\ensuremath{\widetilde{R}}} 
\newcommand{\entity}[1]{\ensuremath{#1}} 
\newcommand{\rel}[1]{\ensuremath{#1}} 
\newcommand{\comp}{\mathbin{;}} 
\newcommand{\princ}[1]{\textsf{#1}} 
\newcommand{\ppmc}{\ensuremath{\phi}} 
\newcommand{\npmc}{\ensuremath{\psi}} 
\newcommand{\pmp}{\ensuremath{\rho}} 
\newcommand{\act}[1]{\textsf{#1}} 
\newcommand{\crs}{\ensuremath{\chi}} 
\newcommand{\crso}[1]{\textsf{#1}} 
\newcommand{\defo}[1]{#1} 
\newcommand{\audita}[1]{\ensuremath{#1^\oplus}}

\newcommand{\auditd}[1]{\ensuremath{#1^\ominus}}

\newcommand{\interesta}{\ensuremath{i^\oplus}}
\newcommand{\interestb}{\ensuremath{i^\ominus}}
\newcommand{\setinterest}{\ensuremath{\set{i^\oplus,i^\ominus}}}

\begin{document}

\title{Relationships, Paths and Principal Matching:\\A New Approach to Access Control}

\author{Jason Crampton}
\author{James Sellwood}
\affil{Information Security Group\\ Royal Holloway, University of London}


\maketitle
\begin{abstract}
  Recent work on relationship-based access control has begun to show how it can be applied to general computing systems, as opposed to simply being employed for social networking applications. The use of relationships to determine authorization policies enables more powerful policies to be defined than those based solely on the commonly used concept of role membership.

  The \underline{r}elationships, \underline{p}aths and \underline{p}rincipal \underline{m}atching (RPPM) model described here is a formal access control model using relationships and a two-stage request evaluation process. We make use of path conditions, which are similar to regular expressions, to define policies. We then employ non-deterministic finite automata to determine which policies are applicable to a request.

  The power and robustness of the RPPM model allows us to include contextual information in the authorization process (through the inclusion of logical entities) and allows us to support desirable policy foundations such as separation of duty and Chinese Wall. Additionally, the RPPM model naturally supports a caching mechanism which has significant impact on request evaluation performance.
\end{abstract}


\section{Introduction}\label{sec:intro}
In modern computing there are very few computer systems where every user of that system is required to be able to perform all possible actions on all possible resources.
More usually there is a need to selectively limit the actions which can be performed on resources; access control is the security service which provides this capability.
The reasons why such limitations are required vary depending on the computer system.
There may simply be a distinction between configuration (administrative) functions versus operation (user) functions, or between a write mode and a read mode.
There may be a need to isolate the data belonging to each individual user from every other user, or there may be more complex requirements based on concepts such as security clearance level, job role or previous activity.

In many situations, access control is policy-based: interactions between users and resources are modeled as ``requests'' and the policy specifies (either implicitly or explicitly) which requests are to be granted and which denied.
An access control system is based on an access control model, which will define the data structures used to specify an authorization policy and an algorithm, often known as a policy decision point, to determine whether a request is authorized by a given policy.
As is customary in the literature we will use the terms \emph{subjects} and \emph{objects} when referring to the parties who are to, respectively, perform and be the target of authorization (inter)actions.

There have been numerous access control models defined since the topic first attracted interest in the mid-1960s.
The protection matrix model simply enumerated all authorized actions~\cite{Lampson_Protection}.
This is a conceptually simple approach, however, it is inefficient when dealing with more than a few subjects and objects.
New models have been introduced with the intention of addressing limitations in existing models or to accommodate richer types of authorization policy.

A prime example of this is the Role-Based Access Control (RBAC) model which allows permissions to perform actions to be granted to job roles.
Subjects are assigned to their applicable roles and thus gain the permissions to perform the necessary actions for that role.
The RBAC model offers several significant benefits over previous models.
In particular, it reduces the administrative burden of managing the access control system by abstracting policy assignment away from subjects to roles; additionally, it is conceptually simple, thereby being easily understood and implemented.
It is principally for these reasons that it (or some close variant) has become so widely utilised in modern computing systems.
Since RBAC's inception there have been numerous variations and extensions suggested to adapt it for specific applications.
These extensions have included, for example, support for role hierarchies, as well as geographical and temporal constraints~\cite{BertinoBF01,DamianiBCP07}.

More recently, alternative models have been growing in popularity, with Attribute-Based Access Control (ABAC)~\cite{Al-KahtaniS02} and Usage Control (UCON)~\cite{SandhuP03} receiving particular attention.
All these models assume that authorization should, essentially, be based on user attributes (particularly user identities).
However, in many computing systems, it is not the individual that is relevant to the access control decision, but the relationship that exists between the individual requesting access and the resource to which access is requested.
Consider, for example, a request by a user $u$ to read the records of a patient $p$.
The fact that $u$ is a doctor is a necessary, but not sufficient, condition for access to be granted.
Specifically, $u$ should be one of $p$'s doctors.
A second example arises when the same user may occupy different roles in different contexts.
A PhD student, for example, may be an enroled student on course $c_1$ and a teaching assistant for course $c_2$.
Clearly, a request to read the coursework of another student should be disallowed if the coursework is for course $c_1$ and allowed if for course $c_2$.
Whilst parameterized variants of RBAC are able to bundle the context into the role~\cite{GiuriI97}, this often leads to a proliferation of roles as each specific context must be `identified'.
As the number of roles tends towards the number of users this undermines RBAC's reduced administrative burden.
Access control languages based on first order logic or logic programming can express complex access control policies that can deal with such situations~\cite{BeckerFG10,GurevishN08}.
However, this comes at the cost of complexity, both for end users that have to specify policies and in terms of policy evaluation.

A new paradigm, known as relationship-based access control, has emerged, particularly to address access control in online social networks~\cite{CarminatiFP09,Fong11}.
In this paper, we extend relationship-based access control to arbitrary computing systems.
We provide a richer policy framework than RBAC, taking relationships into account, while retaining conceptual simplicity.
However, we also exploit features of RBAC and Unix to provide a scalable and intuitive policy language and evaluation strategy.
We introduce the concept of a \emph{path condition}, which is used to associate a request with a set of security principals at request time.
The security principals are authorized to perform particular actions.
Thus, at a high-level a security principal is analogous to a role.

The RPPM model takes its inspiration from the Unix access control model, RBAC and existing work on relationship-based access control.
However, it provides a much richer and more flexible basis for specifying access control policies than any of these models.
In particular, it provides arbitrary flexibility in the definition of principals, unlike Unix; it supports policy specification based on relationships, unlike RBAC; and it provides policy abstraction (based on principals) and support for general-purpose computing systems, unlike existing work on relationship-based access control (which has focused on social networks).
At the heart of the RPPM model is the system graph and the set of relationships.
In this paper, we describe how relationships are used to define principal-matching policies and how requests are interpreted in the context of the system graph and principal-matching policies.
We demonstrate how a policy decision point can be constructed, based on non-deterministic finite automata (NFA).
We also describe how the RPPM model can easily support more advanced access control concepts, which greatly increase the performance and flexibility of the model, thereby broadening its suitability for practical implementations.

In the next section, we formally define the core RPPM model and the authorization policy structures.
Section~\ref{sec:requests} describes how NFAs are used to evaluate requests within the model.
Section~\ref{sec:extended_typed_edges} details a set of extensions to the basic model, by introducing the notion of typed edges.
We provide a discussion of related work in Section~\ref{sec:related-work} and draw conclusions in Section~\ref{sec:conclusion}.

\section{The Model}\label{sec:model}
The central component of the RPPM model is the system graph, a labelled graph in which nodes represent the entities of the system being modelled and the labelled edges represent relationships between them.
As well as concrete entities, such as users and resources, nodes can also represent logical entities with which other entities are associated.
These logical entities can be employed to give context, or some system-specific grounding, to the concrete entities.
For example, in the case of a medical records management system we may have concrete entities representing patients, doctors, healthcare records and medicines; additionally we may have logical entities representing medical cases, healthcare teams and research projects.

A labelled edge linking two nodes identifies a relationship between these nodes.
Such edges may be directed (asymmetric) or undirected (symmetric) depending on the type of relationship.
Paths of edges in the system graph are used to match path conditions which identify principals to be associated with an authorization request.
It is these principals which are assigned permissions to perform actions on objects.

\subsection{System Model and System Graph}\label{sec:model:system_graph}
The RPPM model is designed as a general model for access control, utilising relationship information in order to make authorization decisions.
This generality comes from the model's ability to support whatever entity and relationship types are necessary to describe a particular system at the desired level of detail (unlike Unix, say).
Whilst such flexibility is powerful, it can also limit the checks and controls available for administration of an implementation of the model.
In order to provide an underlying structure and, therefore, a basis on which to incorporate the appropriate checks and controls, we first define a system model which constrains the ``shape'' of the system graph.

\begin{definition}
    A \emph{system model} comprises a set of types $T$, a set of relationship labels $R$, a set of \emph{symmetric} relationship labels $S \subseteq R$ and a \emph{permissible relationship graph} $G_{\textrm{PR}} = (V_{\textrm{PR}},E_{\textrm{PR}})$, where $V_{\textrm{PR}} = T$ and $E_{\textrm{PR}} \subseteq T \times T \times R$.
\end{definition}

\begin{definition}
    Given a system model $(T,R,S,G_{\textrm{PR}})$, a \emph{system instance} is defined by a \emph{system graph} $G = (V,E)$, where $V$ is the set of entities and $E \subseteq V \times V \times R$, and a function $\tau : V \rightarrow T$ which maps an entity to a type.
    We say $G$ is \emph{well-formed} if for each entity $v$ in $V$, $\tau(v) \in T$, and for every edge $(v,v',r) \in E$, $(\tau(v),\tau(v'),r) \in E_{\textrm{PR}}$.
\end{definition}

Our definition of a system graph allows for multiple edges between nodes, as multiple relationships frequently exist between entities in the real world; such a graph is sometimes called a \emph{multigraph}.
The administrative interface for any implementation of the RPPM model must ensure that the system graph is always well-formed with respect to its underlying system model.
As the system being modelled may well be dynamic, updates to the system graph must be controlled in order to continue to maintain its well-formedness.
Additionally, as will be seen in Section~\ref{sec:extended_typed_edges}, there is great utility in supplementing the system graph's relationship edges with ones derived during operation of the authorization system.
Such additions must transform a system graph from one well-formed state to another.
The extensions described in Section~\ref{sec:extended_typed_edges}, introduce specific ``system'' types of edges to the model.

\subsection{Path Conditions}\label{sec:model:path_condition}
In order to limit the administrative burden of defining access control policies in systems with many subjects, the RPPM model abstracts permission assignment to security principals (in the same way that roles simplify policy specification and maintenance in RBAC).
To determine if a particular principal is relevant to a request, an associated chain of relationships must be matched between the subject and object of the request.
Such chains of relationships are called path conditions.
They are composed of relationship labels, organised as sequences, with support for several regular expression-like operators.%
\footnote{The definition of path condition employs common regular expression operators, as do~\cite{ChengPS12dbsec} and~\cite{KhanF12}. We exclude disjunction and the Kleene star for reasons described in Section~\ref{sec:model:policies}.}

\begin{definition}\label{def:path-condition}
   Given a set of relationships $R$, we define a \emph{path condition} recursively:
    \begin{itemize}
        \item $\diamond$ is a path condition;
        \item $r$ is a path condition for all $r \in R$;
        \item if $\pi$ and $\pi'$ are path conditions, then $(\pi)$, $\pi \comp \pi'$, $\pi^+$ and $\overline{\pi}$ are path conditions.
    \end{itemize}
   A path condition of the form $r$ or $\overline{r}$, where $r \in R$, is said to be an \emph{edge condition}.
\end{definition}

Informally, $\pi \comp \pi'$ represents the concatenation of two path conditions; $\pi^+$ represents one or more occurrences, in sequence, of $\pi$; and $\overline{\pi}$ represents $\pi$ reversed; $\diamond$ defines an ``empty'' path condition.
$(\pi)$ provides a means of clearly indicating the extent of path condition $\pi$ such that use of the Kleene plus operator is unambiguous.
The satisfaction of a path condition is defined relative to a system graph $G$ and two nodes $u$ and $v$ in the graph.

\begin{definition}\label{def:path-condition-satisfaction}
    Given a system graph $G = (V,E)$ and $u,v \in V$, we write $G,u,v \models \pi$ to denote that  $G$, $u$ and $v$ \emph{satisfy path condition} $\pi$.
    Then, for all $G,u,v,\pi,\pi'$:
    \begin{itemize}
        \item $G,u,v \models \diamond$ iff $v = u$;
        \item $G,u,v \models r$ iff $(u,v,r) \in E$;
        \item $G,u,v \models (\pi)$ iff $G,u,v \models \pi$;
        \item $G,u,v \models \pi \comp \pi'$ iff there exists $w \in V$ such that $G,u,w \models \pi$ and $G,w,v \models \pi'$;
        \item $G,u,v \models \pi^+$ iff $G,u,v \models \pi$ or $G,u,v \models \pi \comp \pi^+$;
        \item $G,u,v \models \overline{\pi}$ iff $G,v,u \models \pi$.
    \end{itemize}
\end{definition}

The compositional nature of path conditions, along with the regular expression-like operators, means that there is flexibility in how chains of relationships can be specified in a path condition.
For example, the path conditions $\pi \comp \pi^+$ and $\pi^+ \comp \pi$ are valid representations of the same chain of relationships -- specifically, two or more instances of the relationship $\pi$.
We now define what we mean by the equivalence of two path conditions, enabling us to define the concept of \emph{simple} path conditions; henceforth, we will be assume all path conditions are simple.

\begin{definition}\label{def:path-condition-equivalence}
    Path conditions $\pi$ and $\pi'$ are said to be \emph{equivalent}, denoted $\pi \equiv \pi'$, if, for all system graphs $G = (V,E)$ and all $u,v \in V$ we have
    \[
        G,u,v \models \pi \quad\text{if and only if} \quad G,u,v \models \pi'.
    \]
\end{definition}

Trivially, by Definition~\ref{def:path-condition-satisfaction} and the definition of a symmetric relationship, we have
\begin{inparaenum}[(i)]
   \item $\overline{\diamond} \equiv \diamond$;
    \item $(\pi) \equiv \pi$ for all path conditions $\pi$;
   \item $\overline{s} \equiv s$ for all $s \in S$.
\end{inparaenum}
We also have the following results.

\begin{proposition}\label{pro:simple-equivalences}
    For all path conditions $\pi_1$ and $\pi_2$:
    \begin{enumerate}[\em (i)]
        \item $\pi_1 \comp \diamond \equiv \diamond \comp \pi_1 \equiv \pi_1$;
        \item $\overline{\pi_1^+} \equiv \overline{\pi_1}^+$;
        \item $\overline{\overline{\pi_1}} \equiv \pi_1$;
        \item $\overline{\pi_1 \comp \pi_2} \equiv \overline{\pi_2} \comp \overline{\pi_1}$;
        \item $(\pi^+)^+ \equiv \pi^+$;
        \item $\pi_1^+ \comp \pi_1 \equiv \pi_1 \comp \pi_1^+$.
    \end{enumerate}
\end{proposition}

\begin{proof}
    All results follow immediately from Definitions~\ref{def:path-condition-satisfaction} and~\ref{def:path-condition-equivalence}.
    Consider (iv), for example.
    By definition, $G,u,v \models \overline{\pi_1 \comp \pi_2}$ if and only if $G,v,u \models \pi_1 \comp \pi_2$.
    And $G,v,u \models \pi_1 \comp \pi_2$ if and only there exists $w$ such that $G,v,w \models \pi_1$ and $G,w,u \models \pi_2$.
    Thus we have $G,u,v \models \overline{\pi_1 \comp \pi_2}$ if and only if there exists $w$ such that $G,w,v \models \overline{\pi_1}$ and $G,u,w \models \overline{\pi_2}$.
    That is $G,u,v \models \overline{\pi_2} \comp \overline{\pi_1}$.
\end{proof}

\begin{definition}\label{def:simple-path-condition}
   Given a set of relationships $R$, we define a \emph{simple path condition} recursively:
    \begin{itemize}
        \item $\diamond$, $r$ and $\overline{r}$, where $r \in R$, are simple path conditions;
        \item if $\pi \ne \diamond$ and $\pi' \ne \diamond$ are simple path conditions, then $(\pi)$, $\pi \comp \pi'$ and $\pi^+$ are simple path conditions.
    \end{itemize}
\end{definition}

In other words, $\overline{\star}$ occurs in a simple path condition if and only if $\star$ is an element of $R$.
It follows from Proposition~\ref{pro:simple-equivalences} that every path condition may be reduced to a simple path condition.
The path condition $\overline{\overline{r_1 \comp r_2} \comp (r_1 \comp r_3)^+}$, for example, can be transformed into the equivalent path condition $(\overline{r_3} \comp \overline{r_1})^+ \comp r_1 \comp r_2$ using the equivalences in Proposition~\ref{pro:simple-equivalences}.

\begin{remark}\label{rem:r-includes-rbar}
    Henceforth, we assume all path conditions are simple.
    Thus we may define the set of relationship labels to be $\RcoR = R \cup \overline{R}$, where $\overline{R}$ is defined to be $\set{\overline{r} : r \in R}$.
    Given this formulation, the system graph must satisfy the following consistency requirements:
    \begin{itemize}
     \item $(t,t',r) \in E_{\textrm{PR}}$ if and only if $(t',t,\overline{r}) \in E_{\textrm{PR}}$;
     \item $(v,v',r) \in E$ if and only if $(v',v,\overline{r}) \in E$;
     \item $(v,v',s) \in E$ if and only if $(v',v,s) \in E$.
    \end{itemize}
\end{remark}

\subsection{Policies}\label{sec:model:policies}

\newcommand{\all}{\mathsf{all}}
\newcommand{\none}{\mathsf{none}}

The RPPM model employs two policies: the principal-matching policy and the authorization policy.
The purpose of the principal-matching policy is to determine which principals are relevant to an access request.
The purpose of the authorization policy is to determine the actions for which a principal is authorized.

An authorization principal is, therefore, a central component in RPPM policies.
A request is mapped to a set of principals and each principal is authorized to perform particular actions.
In other words, principals are analogous to roles in a role-based access control (RBAC) model.
However, the way in which principals are associated with subjects is very different from RBAC.

A second important component is a \emph{target}.
Each principal-matching rule specifies two targets.
Every path condition $\pi$ is a target and we say a request $(s,o,a)$, where $s$ and $o$ are vertices in the system graph $G$ (and $a \in A$ is a requested action), \emph{matches} target $\pi$ if $G,s,o \models \pi$.
We define two special targets: $\all$ and $\none$, where $\all$ matches every request and $\none$ matches no request.
By a slight abuse of notation ($\all$ and $\none$ are not path conditions), we will write $G,s,o \models \all$ and $G,s,o \not\models \none$, for any request $(s,o,a)$.
Given a request $q = (s,o,a)$, where $s$ and $o$ are vertices in the system graph $G$, we will write $G,q \models \pi$, rather than $G,s,o \models \pi$, to simplify notation.

\begin{definition}
    Let $P$ be a set of authorization principals.
    A \emph{principal-matching rule} has the form $(\ppmc,\npmc,p)$, where $p \in P$ and $\ppmc$ and $\npmc$ are targets.
    A \emph{principal-matching policy} is a set of principal-matching rules.
\end{definition}

Informally, targets are used to determine which rules are applicable to a given request, where $\ppmc$ specifies a required path in the system graph and $\npmc$ specifies a forbidden path.
The meaning of a principal-matching policy (PMP) is defined in the context of a system graph and a request.

\begin{definition}\label{def:matched-principals}
    We say a principal-matching rule $(\ppmc,\npmc,p)$ is \emph{applicable} to a request $q = (s,o,a)$ if and only if $G,q \models \ppmc$ and $G,q \not\models \npmc$.
    Given a system graph $G = (V,E)$, a PMP $\rho$
    and a request $q = (s,o,a)$, where $s,o \in V$, we define the set of \emph{matched principals}:
    \[
    \semantics{\rho}^G_q \stackrel{\rm def}{=} \set{p \in P : (\ppmc,\npmc,p) \in \rho\ \text{is applicable to $q$}}.
    \]
\end{definition}

Henceforth, $G$ will be assumed to be given, so we will simply write $\semantics{\rho}_q$ to denote the set of matched principals for policy $\rho$ and request $q$.
We will further abbreviate this to $\semantics{\rho}$ when $q$ is obvious from context.

\begin{remark}
 Path conditions are clearly closely related to regular expressions.
 However, we do not include disjunction or the Kleene star operator in our definition of path conditions.
 Instead, we use two (or more) principal-matching rules.
 The path condition $\pi^* \comp \pi'$, for example, can be associated with a principal $p$ by specifying the principal-matching rules $(\pi',\none,p)$ and $(\pi^+ \comp \pi',\none,p)$.
 This approach is preferable to including these features as they're inclusion would not increase the expressiveness of the policy language but would introduce greater complexity into the request evaluation process discussed in Section~\ref{sec:requests:path-conditions-to-principal-matching}\footnote{Specifically we would require an additional NFA construction mechanism for each.}.
\end{remark}

\begin{remark}
 A PMP may specify a default principal $p_{\rm def}$, much like the concept of ``world'' in the Unix access control system.
 To do so, we include the principal-matching rule $(\all,\none,p_{\rm def})$.
\end{remark}

\begin{definition}\label{def:auth-decisions}
  An \emph{authorization rule} has the form $(p,x,y,b)$, where $p \in P$, $x \in O \cup T \cup \set{\star}$, $y \in A \cup \set{\star}$ and $b \in \set{0,1}$.\footnote{Recall $T$ is the set of entity types.}
  Given a PMP $\rho$, an authorization rule $(p,x,y,b)$ is applicable to a request $q = (s,o,a)$ if all of the following conditions hold:
  \begin{itemize}
   \item $p \in \semantics{\rho}$;
   \item $x \in \set{o,\tau(o),\star}$;
   \item $y \in \set{a,\star}$.
  \end{itemize}
    An \emph{authorization policy} is a set of authorization rules.
    Given an authorization policy $\varrho$
    and a request $q = (s,o,a)$, we define the set of \emph{authorization decisions}:
    \[
    \semantics{\rho,\varrho}^G_q \stackrel{\rm def}{=} \set{b \in \set{0,1} : (p,x,y,b) \in \varrho\ \text{is applicable to}\ q}.
    \]
\end{definition}

The rule $(p,o,a,1)$ indicates that principal $p$ is authorized to perform action $a$ on object $o$, while $(p,o,a,0)$ indicates $p$ is not authorized.
The wild card character $\star$ is used to simplify the specification of policies.
It can be used, for example, to authorize a principal for all actions on a given object.
We can then prohibit specific actions using a negative authorization tuple.
Thus, the inclusion of $(p,o,\star,1)$ and $(p,o,a,0)$ in the policy authorizes $p$ for all actions on object $o$, \emph{except} action $a$.

The ability to specify object authorization rules for individual objects $o$ or all objects $\star$ provides for flexible policy creation.
However, in some circumstances these two extremes may not be appropriate.
Support for authorization rules specified in terms of object types goes some way to balancing these two alternatives.
The rule $(p,t,a,1)$ indicates that principal $p$ is authorized to perform action $a$ on all objects of type $t$, whilst $(p,t,a,0)$ indicates $p$ is not authorized.

Again, we will simply write $\semantics{\rho,\varrho}_q$ to indicate the set of authorization decisions for policy $\varrho$ and further abbreviate this to $\semantics{\rho,\varrho}$ when no ambiguity will arise.

\begin{example}\label{ex:higher-education-1}
Returning to our higher education example from Section~\ref{sec:intro}, we can envisage a system that includes a PhD student $u_1$ and professor $u_2$, courses $c_1$ and $c_2$, coursework answers $a_1$ and $a_2$ for course $c_1$, and coursework answer $a_3$ for course $c_2$ (illustrated by the system graph fragment shown in Figure~\ref{img:he-system-graph-fragment}).

\begin{figure}[!ht]\centering\setlength{\extrarowheight}{2pt}
         \begin{tikzpicture}
            [node distance=1.5cm and 3cm, on grid,
            caching/.style={color=blue!70,>=open diamond}, 
            audit-a/.style={color=teal, densely dashed}, 
            audit-d/.style={color=red, densely dotted}, 
            interest-a/.style={color=purple,>=\interestahead}, 
            interest-b/.style={color=brown,>=\interestdhead}, 
            every circle node/.style={draw,minimum width=20pt},thick,
            every node/.append style={scale=0.7, transform shape}]
            \begin{scope}[>=latex] 
                \node[circle] (u1) {\entity{u_1}};
                \node[circle,left=of u1] (c1) {\entity{c_1}};
                \node[circle,right=of u1] (c2) {\entity{c_2}};
                \node[circle,left=of c1] (a1) {\entity{a_1}};
                \node[circle,right=of c2] (a3) {\entity{a_3}};
                \node[circle,below=of c1] (a2) {\entity{a_2}};
                \node[circle,above=of c1] (u2) {\entity{u_2}};
                \draw[thick,->] (u1) to node[swap,auto] {\textsf{is-enrolled-on}} (c1);
                \draw[thick,->] (u2) to node[auto] {\textsf{is-responsible-for}} (c1);
                \draw[thick,->] (u1) to node[auto] {\textsf{is-ta-for}} (c2);
                \draw[thick,->] (a1) to node[auto] {\textsf{is-coursework-for}} (c1);
                \draw[thick,->] (a3) to node[swap,auto] {\textsf{is-coursework-for}} (c2);
                \draw[thick,->] (a2) to node[auto] {\textsf{is-coursework-for}} (c1);
                \draw[thick,->] (u1) to node[auto] {\textsf{is-creator-of}} (a2);
              \end{scope}
      \end{tikzpicture}
      \caption{A fragment of the system graph for the higher education use case}\label{img:he-system-graph-fragment}
\end{figure}
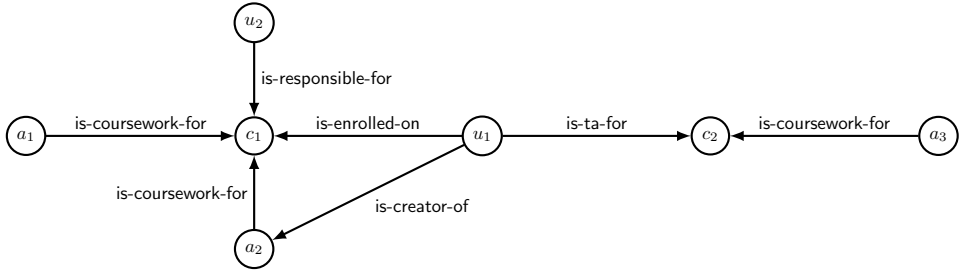

It is natural that we would wish to constrain the PhD student from accessing answers (other than their own) for courses on which they are enrolled as a student, whilst we would wish to grant access to those for courses for which they are a teaching assistant.
To do so requires the ability to distinguish the context (in this case the course) associated with a request.
We can achieve this in the RPPM model through the specification of the following policies:
  \begin{align*}
   \rho = \{&(\textsf{is-creator-for},\none,\princ{author}), \\
        & (\textsf{is-ta-for} \comp \overline{\textsf{is-coursework-for}}, \textsf{is-enrolled-on} \comp \overline{\textsf{is-coursework-for}},\princ{course-ta}), \\
        & (\textsf{is-responsible-for}  \comp \overline{\textsf{is-coursework-for}},\none,\princ{course-leader})\} \\
   \varrho = \{&(\princ{author},\star,\act{read},1),
         (\princ{author},\star,\act{write},1),
         (\princ{course-ta},\star,\act{read},1),
         (\princ{course-ta},\star,\act{grade},1), \\
        & (\princ{course-leader},\star,\act{read},1),
        (\princ{course-leader},\star,\act{review},1) \}
  \end{align*}

We will consider request evaluation in more detail in Section~\ref{sec:requests}.
However, the intuition behind request evaluation is that we determine whether, for a given request and system graph, there is a path in the graph from subject to object for which the associated labels match the path condition.
Consider the request $(u_1, a_1, \act{read})$, for example.
There is no path in the graph between $u_1$ and $a_1$ that matches any of the mandated targets in rules within $\rho$.
Thus, the set of matched principals is empty (which will lead to the request being denied, assuming a deny-by-default discipline).
On the other hand, the set of matched principals for request $(u_1, a_2, \act{read})$ is $\set{\princ{author}}$, since there is a path from $u_1$ to $a_2$ with label \textsf{is-creator-of}; hence the request will be granted (because of the first rule in $\varrho$).
However, the set of matched principals for request $(u_1,a_3,\act{read})$ is $\set{\princ{course-ta}}$ and the request will be permitted (because of the third rule in $\varrho$).
Note the difference in outcomes for requests $(u_1,a_1,\act{read})$ and $(u_1,a_3,\act{read})$ because of the different relationships that exist between $u_1$ and the courses associated with objects $a_1$ and $a_3$.
Similarly, the set of matched principals for $(u_2,a_1,\act{read})$ and $(u_2,a_2,\act{read})$ is $\set{\princ{course-leader}}$ and these requests will be granted, whereas the set of matched principals for $(u_2,a_3,\act{read})$ is empty (and the request will be denied).
Again, the professor's relationship with the two courses determines the  principals (and thus decisions) associated with the respective requests.
\end{example}

\subsection{Policy Extensions}\label{sec:model:policy-extensions}

We now describe additional refinements of principal-matching policies and authorization policies.
Most importantly, we indicate how we deal with a set of decisions that is not a singleton and how we can provide support for principal activation rules.

\subsubsection{Conflict Resolution}\label{sec:model:policy-extensions:conflict-resolution}
The authorization rule $(p,o,a,0)$ explicitly disallows $p$ from performing action $a$ on object $o$, while $(p,o,a,1)$ explicitly allows it.
Thus, the set of applicable authorization decisions may contain conflicting decisions.

Accordingly, we define an \emph{extended} authorization policy to be a pair $(\varrho,\chi)$, where $\varrho$ is a set of authorization rules and $\chi$ is a \emph{conflict resolution strategy} (CRS) which is used to reduce the set of matching decisions to a single decision.
That is, \mbox{$\semantics{\rho,(\varrho,\chi)} \in \set{\set{0},\set{1}}$}.
In the interests of brevity, we will continue to write $\semantics{\rho,\varrho}$ in preference to $\semantics{\rho,(\varrho,\chi)}$.

The \crso{DenyOverrides} CRS reduces the set to a single deny ($0$) decision if there is at least one $0$ in the set of matching authorization decisions.
Conversely, \crso{AllowOverrides} reduces the set to a single allow decision if there is at least one $1$ in the set.

\subsubsection{List-oriented Policies}\label{sec:model:policy-extensions:list-oriented-policies}
The meanings of a principal-matching policy and an authorization policy are defined in terms of sets.
A number of access control systems are list-oriented, in the sense that the first applicable decision is used, the Unix access control mechanism being one example.
We could equally well require that rules in PMPs and authorization policies are evaluated in a particular order.
(That is PMPs and authorization policies should be lists, rather than sets, of rules.)

In this case, it would make sense to introduce list-oriented processing.
In particular, we might insist that we take the first matched principal, so the meaning of a PMP becomes a single principal.
Similarly, we might insist that we take the first matching authorization decision (so we would not require conflict resolution).

If a list-oriented approach is employed, then the default principal-matching rule $(\all,\none,p_{\rm def})$ must be placed at the end of the list of principal-matching rules.

\subsubsection{Graph- and Tree-based Policies}\label{sec:model:policy-extensions:tree-based-policies}

\newcommand{\sbpmp}{\ensuremath{_{\rm PMP}}}

Principal-matching rules make use of two targets (one mandated and one precluded).
These targets are evaluated when deciding whether a request is to be permitted or not.
As we have seen, we can support disjunction, where a principal is activated if at least one of several path conditions is satisfied, through the use of multiple principal-matching rules for the same principal.
There may also be times when there is a need to match a security principal only if each one of several path conditions is satisfied.
The basic RPPM policy model described above does not support this requirement.

Hence, we introduce the idea of a \emph{policy graph}.
We arrange the rules in a PMP as a directed acyclic graph, making the process of matching principals more like the evaluation of XACML policies.
More formally, a policy graph is a directed acyclic PMP graph $G\sbpmp = (V\sbpmp,E\sbpmp)$.
The PMP graph is required to have a unique node of in-degree $0$, which we will call the \emph{root}.
Each vertex in the PMP graph is a PMP rule.
It is convenient, in this setting, to define a {\sf null} principal; the {\sf null} principal must not appear in any authorization rules.

Evaluation of the PMP graph is performed by a breadth-first search.
A vertex $v$ (that is, a PMP rule) is only evaluated if the request is applicable to each PMP rule on every path from the root to $v$.
(Of course, if we insist that the PMP graph is a tree, there is only one such path.)
As before, the set of matched principals is simply the set of principals associated with applicable rules.

It is easy to see that our list-oriented policies can be represented in this way.%
 \footnote{We define a tree with root node $(\all,\none,{\sf null})$ and each PMP rule is a child of the root node.}
However, this graph-based approach also makes it possible to encode \emph{principal activation rules} of the form ``if $p$ is applicable to a given request then so is principal $p'$''.
Moreover, we can insist that a principal $p$ is only activated if multiple path conditions $\pi_1,\dots,\pi_n$ are satisfied.

Consider the simple policy graph in Figure~\ref{fig:graph-based-pmp}.
Then (ignoring the {\sf null} principal which has no authorizations) the set of matched principals may be one of $\emptyset$, $\set{p_1}$, $\set{p_2,p_4}$, or $\set{p_1,p_2,p_3,p_4}$.
In particular, $p_4$ is activated if $p_2$ is, because the path conditions associated with $p_2$ are satisfied (and the targets associated with $p_4$ are trivially satisfied); and if both $p_1$ and $p_2$ are activated (because their respective path conditions are satisfied) then so is $p_3$ (as well as $p_4$).

\begin{figure}[!ht]\centering
    \begin{tikzpicture}[->, node distance=1.5cm and 2.5cm, on grid]
        \node (a) {$(\all,\none,{\sf null})$};
        \node[below left=of a] (b) {$(\phi_1,\psi_1,p_1)$};
        \node[below right=of a] (c) {$(\phi_2,\psi_2,p_2)$};
        \node[below left=of c] (d) {$(\all,\none,p_3)$};
        \node[below right=of c] (e) {$(\all,\none,p_4)$};
        \draw (a) -- (b);
        \draw (b) -- (d);
        \draw (a) -- (c);
        \draw (c) -- (d);
        \draw (c) -- (e);
    \end{tikzpicture}
    \caption{A graph-based PMP}\label{fig:graph-based-pmp}
\end{figure}
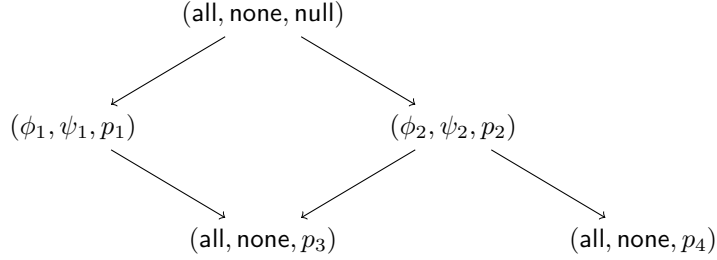

\subsection{Default Decisions}\label{sec:model:defaults}
A default access control decision (allow or deny) needs to be specified in the event that no authorization rules apply to a request.
Systems may need to support \defo{allow-by-default} when the system enters an emergency state, such as the opening of fire exit doors when there is a fire.
Other circumstances will commonly require fail-safe handling, where a \defo{deny-by-default} strategy is implemented in order to ensure no unauthorised access is allowed. Some systems may be deemed so sensitive that there may be no conditions under which \defo{allow-by-default} would be enabled.

There are two circumstances in the RPPM model when default decision-making applies.
The first is when no matched principals are identified (\mbox{$\semantics{\rho} = \emptyset$}), whilst the second is when there are no explicit authorisations (\mbox{$\semantics{\rho,\varrho} = \emptyset$}).
Accordingly, we allow for default decisions to be applied at one of the following levels: default-per-subject, default-per-object, default-per-type or system-wide default.
The default-per-subject decision is only applied when there are no matched principals.%
\footnote{It is not applied when there are no explicit authorizations: when the set of possible decisions is determined, the subject is no longer  relevant, having been used to identify the appropriate matched principals.}

The four defaults are evaluated in order, where specified, with the first applicable default determining the authorization decision.
In this way, if there is a default specified for the subject $s$ of the request $(s,o,a)$, the subject's default (allow or deny) applies.
If no subject default is defined for $s$, then the default for the object $o$ of the request shall apply, if specified.
If there is no subject default for $s$ and no object default for $o$, then the default for the type of object $\tau(o)$ shall apply.
If none of these defaults are defined, then the system-wide default shall apply.
Defaults for the subject, object and type are optional and need not be specified for the entities involved in the request.
However, a system-wide default must be specified in order to ensure an authorization decision can be made for every request.

\section{Request Evaluation}\label{sec:requests}
Request evaluation uses a two-step process, as shown in Figure~\ref{img:overview}, where first we \emph{compute principals} and subsequently \emph{compute authorizations}.
This two-step request evaluation process is inspired by Unix, which first determines the relevant principal (from ``owner'', ``group'' and ``world'') and then authorizations (from the permission mask of the object).

\begin{figure}[!ht]\centering
    \begin{tikzpicture}
            [node distance=3cm and 1cm,
            every rectangle node/.style={draw,minimum width=30pt, minimum height=30pt, align=center,inner sep=6pt},thick,
            every node/.append style={scale=0.7, transform shape}]
            \begin{scope}[>=latex] 
                \node[rectangle] (a) {Request};
                \node[rectangle,right=of a] (b) {Compute\\Principals};
                \node[rectangle,right=of b] (c) {Compute\\Authorizations};
                \node[rectangle,right=of c] (d) {Decision};
                \draw[->,thick] (a) to (b);
                \draw[->,thick] (b) to (c);
                \draw[->,thick] (c) to (d);
            \end{scope}
    \end{tikzpicture}
    \caption{Processing overview}\label{img:overview}
\end{figure}
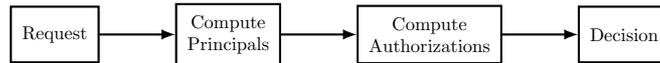

Figure~\ref{img:architecture} provides a detailed architecture of the complete request evaluation process, indicating the inputs necessary and decisions employed; the figure includes the conflict resolution extension described in Section~\ref{sec:model:policy-extensions:conflict-resolution}.
The first step of request evaluation, compute principals, is rather complex and, conceptually, requires the identification of paths within the system graph in order to determine the principals which match a request.
However, the second step, compute authorizations, involves simple lookups to determine whether the matched principals for a request are authorized to perform the requested action on the object.

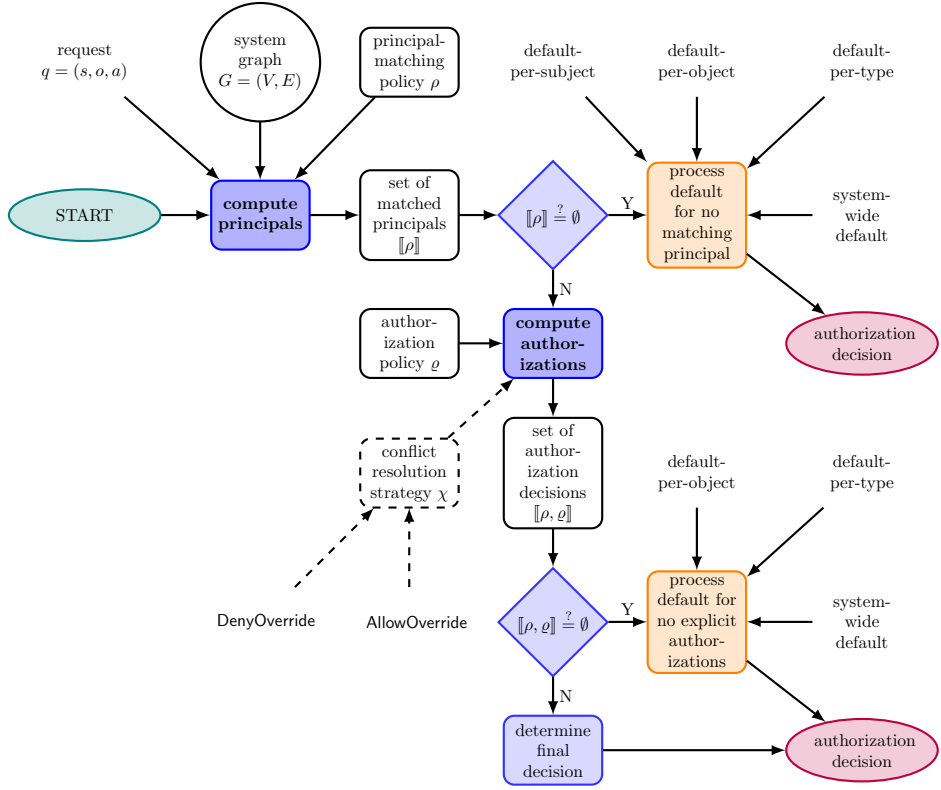
\begin{figure*}[!ht]\centering
  \begin{tikzpicture}
        [auto,
        input/.style={rectangle, text width=5em, align=center, minimum height=4em}, 
        decision/.style={diamond, draw=blue!80, thick, fill=blue!15, text width=4.5em, align=flush center, inner sep=1pt}, 
        block/.style ={rectangle, draw=blue!80, thick, fill=blue!15, text width=5em, align=center, rounded corners, minimum height=4em}, 
        default/.style ={rectangle, draw=orange, thick, fill=orange!20, text width=5em, align=center, rounded corners, minimum height=4em}, 
        coreblock/.style ={rectangle, draw=blue, thick, fill=blue!30, text width=5em, align=center, rounded corners, minimum height=4em}, 
        start/.style ={ellipse, draw=teal, thick, fill=teal!20, text width=5.5em, align=center, minimum height=3em}, 
        end/.style ={ellipse, draw=purple, thick, fill=purple!20, text width=5.5em, align=center, minimum height=3em}, 
        graph/.style ={circle, draw, thick, text width=5em, align=center, minimum height=4em}, 
        list/.style ={rectangle, draw, thick, text width=5em, align=center, minimum height=4em}, 
        set/.style ={rectangle, draw, thick, text width=5em, align=center, rounded corners, minimum height=4em}, 
        strategy/.style ={rectangle, draw, thick, text width=5em, align=center, rounded corners, minimum height=4em, dashed}, 
        line/.style ={draw, thick, -latex}, 
        option/.style ={dashed}, 
        every node/.append style={scale=0.65, transform shape}] 
        \matrix [column sep=5mm,row sep=5mm]
        {
            \node [input] (a1) {request $q = (s,o,a)$}; &
            \node [graph] (a2) {system graph $G = (V,E)$}; &
            \node [set] (a3) {principal-matching policy $\pmp$}; &
            \node [input] (a4) {default-per-subject}; &
            \node [input] (a5) {default-per-object}; &
            \node [input] (a6) {default-per-type};
            \\
            \node [start] (b1) {START}; &
            \node [coreblock] (b2) {\bf compute principals}; &
            \node [set] (b3) {set of matched principals $\semantics{\pmp}$}; &
            \node [decision] (b4) {$\semantics{\pmp} \stackrel{?}{=} \emptyset$}; &
            \node [default] (b5) {process default for no matching principal}; &
            \node [input] (b6) {system-wide default};
            \\
            &
            &
            \node [set] (c3) {author\-ization policy $\varrho$}; &
            \node [coreblock] (c4) {\bf compute author\-izations}; &
            &
            \node [end] (c6) {authorization decision};
            \\
            &
            &
            \node [strategy] (f3) {conflict resolution strategy $\crs$}; &
            \node [set] (d4) {set of author\-ization decisions $\semantics{\pmp, \varrho}$}; &
            \node [input] (d5) {default-per-object}; &
            \node [input] (d6) {default-per-type};
            \\
            &
            \node [input] (e2) {\crso{DenyOverride}}; &
            \node [input] (g2) {\crso{AllowOverride}}; &
            \node [decision] (e4) {$\semantics{\pmp, \varrho} \stackrel{?}{=} \emptyset$}; &
            \node [default] (e5) {process default for no explicit author\-izations}; &
            \node [input] (e6) {system-wide default};
            \\
             &
            &
            &
            \node [block] (f4) {determine final decision}; &
            &
            \node [end] (f6) {authorization decision};
            \\
        };
        \begin{scope}[every path/.style=line]
            \path (a1) -- (b2);
            \path (a2) -- (b2);
            \path (a3) -- (b2);
            \path (a4) -- (b5);
            \path (a5) -- (b5);
            \path (a6) -- (b5);
            \path (b1) -- (b2);
            \path (b2) -- (b3);
            \path (b3) -- (b4);
            \path (b4) -- node {Y} (b5);
            \path (b4) -- node {N} (c4);
            \path (b6) -- (b5);
            \path (b5) -- (c6);
            \path (c3) -- (c4);
            \path (c4) -- (d4);
            \path (d4) -- (e4);
            \path (d5) -- (e5);
            \path (d6) -- (e5);
            \path [option] (e2) -- (f3);
            \path (e4) -- node {Y} (e5);
            \path (e4) -- node {N} (f4);
            \path (e6) -- (e5);
            \path (e5) -- (f6);
            \path [option] (f3) -- (c4);
            \path (f4) -- (f6);
            \path [option] (g2) -- (f3);
        \end{scope}
    \end{tikzpicture}
    \caption{Detailed architecture}\label{img:architecture}
\end{figure*}

Pseudo-code for the entire request evaluation process is shown in Algorithm~\ref{alg:request-eval}.
The \textsf{ApplyDefaults} functionality can be inferred directly from the default decision handling discussion in Section~\ref{sec:model:defaults}, so no formal algorithm is required or provided here.

\begin{algorithm}[t]
\footnotesize
    \caption{\textsf{RequestEvaluation}}
    \label{alg:request-eval}
    \algsetup{indent=2em}
    \begin{algorithmic}[1]
        \REQUIRE{System graph $G = (V,E)$, set of relationship labels $\RcoR$, request $(s,o,a)$, principal-matching policy $\rho$ and extended authorization policy $(\varrho,\chi)$}
        \ENSURE{Returns authorization decision}
            \STATE{$\semantics{\rho} = $ \textsf{ComputePrincipals($G, \RcoR, (s,o,a), \rho$)}}
            \IF{$\semantics{\rho} = \emptyset$}
                \RETURN{\textsf{ApplyDefaults($s,o$)}}
            \ELSE
                \STATE{$\semantics{\rho,\varrho} = $ \textsf{ComputeAuthorizations($(s,o,a), (\varrho,\chi), \semantics{\rho}$)}}
                \IF{$\semantics{\rho,\varrho} = \emptyset$}
                    \RETURN{\textsf{ApplyDefaults($o$)}}
                \ELSIF{$\semantics{\rho,\varrho} = \set{0}$}
                    \RETURN false \COMMENT{deny}
                \ELSIF{$\semantics{\rho,\varrho} = \set{1}$}
                    \RETURN true \COMMENT{allow}
                \ENDIF
            \ENDIF
    \end{algorithmic}
\end{algorithm}

We determine the set of matched principals using Algorithm~\ref{alg:compute-principals} (\textsf{ComputePrincipals}).
The required processing is described in more detail in the next section.
Briefly, we exploit the correspondence between path conditions and regular expressions to build a non-deterministic finite automaton $M_{\pi}$ for path condition $\pi$; and we use the correspondence between labelled graphs and non-deterministic finite automata to construct a non-deterministic finite automaton $M_{G,q}$ derived from the system graph $G$ and information in a request $q$.%
\footnote{In a preliminary version of this paper we made use of a (modified) breadth-first search algorithm to determine whether a target (from within the principal-matching rule) was matched within the system graph~\cite{CramptonS14}.  In this paper, we make use of the rich theory underpinning regular expressions and finite automata to provide an alternative algorithm based on rigorous foundations.}
For brevity, we will write $M_q$ for $M_{G,q}$, as $G$ will always be obvious from context.
We use these non-deterministic finite automata to determine whether each principal-matching rule is applicable to a request, and whether its principal is, therefore, matched.

\begin{algorithm}[t]
\footnotesize
    \caption{\textsf{ComputePrincipals}}
    \label{alg:compute-principals}
    \algsetup{indent=2em}
    \begin{algorithmic}[1]
        \REQUIRE{System graph $G = (V,E)$, set of relationship labels $\RcoR$, request $(s,o,a)$ and principal-matching policy $\rho$}
        \ENSURE{Returns set of matched principals $\semantics{\rho}$}
            \STATE{$\semantics{\rho} = \emptyset$}
            \STATE{$M_{q} = (V, \RcoR, E, s, \set{o})$}
            \FOR{$(\ppmc,\npmc,p) \in \rho$} 
                        \IF{$(\ppmc = \all$) \OR $(L(M_\ppmc) \cap L(M_q) \neq \emptyset)$}
                            \IF{($\npmc = \none$) \OR $(L(M_\npmc) \cap L(M_q) = \emptyset)$}
                                \STATE{$\semantics{\rho} = \semantics{\rho} \cup p$}
                            \ENDIF
                        \ENDIF
            \ENDFOR
    \end{algorithmic}
\end{algorithm}

Having identified the set of matched principals, the set of authorizations can be easily determined from the applicable authorization rules (as per Definition~\ref{def:auth-decisions}) using Algorithm~\ref{alg:compute-authorizations} (\textsf{ComputeAuthorizations}).
This process is far simpler than that of Algorithm~\ref{alg:compute-principals}, being limited to simple comparisons and set membership checks.
Lines 7 through 11 of Algorithm~\ref{alg:compute-authorizations} provide for conflict resolution as described in Section~\ref{sec:model:policy-extensions:conflict-resolution}.

\begin{algorithm}[t]
\footnotesize
    \caption{\textsf{ComputeAuthorizations}}
    \label{alg:compute-authorizations}
    \algsetup{indent=2em}
    \begin{algorithmic}[1]
        \REQUIRE{Request $(s,o,a)$, extended authorization policy $(\varrho,\chi)$ and set of matched principals $\semantics{\rho}$}
        \ENSURE{Returns set of authorization decisions $\semantics{\rho,\varrho}$}
            \STATE{$\semantics{\rho,\varrho} = \emptyset$}
            \FOR{$(p,x,y,b) \in \varrho$} 
                \IF{$(p \in \semantics{\rho})$ \AND $((x = o) \OR (x = \tau(o)) \OR (x = \star))$ \AND $((y = a) \OR (y = \star))$}
                    \STATE{$\semantics{\rho,\varrho} = \semantics{\rho,\varrho} \cup b$}
                \ENDIF
            \ENDFOR
            \IF{($\chi = \crso{DenyOverrides})$ \AND $(0 \in \semantics{\rho,\varrho})$}
                \STATE{$\semantics{\rho,\varrho} = \set{0}$}
            \ELSIF{$(\chi = \crso{AllowOverrides})$ \AND $(1 \in \semantics{\rho,\varrho})$}
                \STATE{$\semantics{\rho,\varrho} = \set{1}$}
            \ENDIF
    \end{algorithmic}
\end{algorithm}

Recall that we may define more complex graph-based policies, in which the principal-matching rules are arranged in an acyclic directed graph.
Accordingly, we would need to modify the request evaluation algorithm to ensure that we identify all matched principals.
Informally, this means the \textsf{ComputePrincipals} algorithm becomes a breadth-first search of the policy graph which calls a sub-routine for testing path conditions at each node visited.

\subsection{From Path Conditions to NFAs to Principal Matching}\label{sec:requests:path-conditions-to-principal-matching}

We now describe the correspondence between path conditions and non-deterministic finite automata in more detail.
We also explain how to define the automaton $M_q$ given a system graph $G$ and a request $q$.
Finally, we explain how to construct an automaton that will determine whether a request $q$ matches a path condition $\pi$ (in the context of a system graph $G$).

A \emph{non-deterministic finite automaton} is a 5-tuple $M = (Q, \Sigma, \delta, s, F)$ where:
\begin{itemize}
 \item $Q$ is the set of \emph{states},
 \item $\Sigma$ is the set of inputs (the \emph{alphabet}),
 \item $\delta \subseteq Q \times Q \times \Sigma$ is the \emph{transition relation},
 \item $s \in Q$ is a \emph{start state} and $F \subseteq Q$ is the set of \emph{accepting states}.
\end{itemize}
Let $\omega = \sigma_1 \dots \sigma_\ell$, where $\sigma_i \in \Sigma$, be a \emph{word} over the alphabet $\Sigma$.
The NFA $M$ \emph{accepts} word $\omega$ if there exists a sequence of states, $q_0,\dots,q_\ell$ such that:
\begin{itemize}
 \item $s = q_0$ and $q_i \in Q$ for $i > 0$;
 \item $(q_i, q_{i+1},\sigma_{i+1}) \in \delta$ for $0 \leq i \leq \ell-1$;
 \item $q_\ell \in F$.
\end{itemize}
We write $L(M)$ to denote the set of words (or \emph{language}) accepted by $M$.

Given two NFAs, \mbox{$M_1 = (Q_1, \Sigma_1, \delta_1, s_1, F_1)$} and \mbox{$M_2 = (Q_2, \Sigma_2, \delta_2, s_2, F_2)$}, the \emph{intersection NFA} \mbox{$M_\cap = (Q_1 \times Q_2, \Sigma_1 \cap \Sigma_2, \delta_\cap, (s_1, s_2), F_1 \times F_2)$} accepts the language $L(M_1) \cap L(M_2)$, where
\[
\delta_\cap = \set{((q_1,q_2),(q_1',q_2'),\sigma) : (q_1,q_1',\sigma) \in \delta_1, (q_2,q_2',\sigma) \in \delta_2}.
\]

\subsubsection{Path Conditions as NFAs}

We now explain how to construct an NFA for a path condition.
The construction is straightforward and is based on standard techniques (see, for example,~\cite{Aho_Compilers}), given the obvious similarities between path conditions and regular expressions.

\begin{proposition}
Let $r \in \RcoR$ and let $\pi$ and $\phi$ be path conditions with NFAs $M_{\pi} = (Q_\pi,\Sigma_\pi,\delta_\pi,s_\pi,\set{f_\pi})$ and $M_{\phi}  = (Q_\phi,\Sigma_\phi,\delta_\phi,s_\phi,\set{f_\phi})$ accepting languages $L(M_\pi)$ and $L(M_\phi)$, respectively. Then 
\begin{itemize}
 \item \mbox{$M_r = (\set{s,f}, \set{r}, \set{(s,f,r)}, s, \set{f})$};
 \item $M_{\pi \comp \phi} = (Q_{\pi \comp \phi},\Sigma_\pi \cup \Sigma_\phi,\delta_{\pi \comp \phi},s_\pi,\set{f_\phi})$, where $Q_{\pi \comp \phi} = Q_\pi \cup Q_\phi \setminus \set{s_\phi}$ and \[ \delta_{\pi \comp \phi} = \delta_\pi \cup \delta_\phi \cup \set{(f_\pi,q,r) : (s_\phi,q,r) \in \delta_\phi} \setminus \set{(x,y,z) \in \delta_\phi : x = s_\phi}; \]
 \item \mbox{$M_{\pi^+} = (Q_\pi,\Sigma_\pi \cup \set{\epsilon},\delta_{\pi^+},s_\pi,\set{f_\pi})$}, where $\epsilon$ is the empty symbol and \[ \delta_{\pi^+} = \delta_\pi \cup \set{(f_\pi,s_\pi,\epsilon)}. \]
\end{itemize}
\end{proposition}

By construction, every NFA will have a single final state.
Moreover, because we do not include disjunction in path conditions, there is a unique transition from the initial state.
The constructions of $M_r$, $M_{\pi \comp \phi}$ and $M_{\pi^+}$ are illustrated in Figure~\ref{img:nfa-constructions-with-epsilon}.

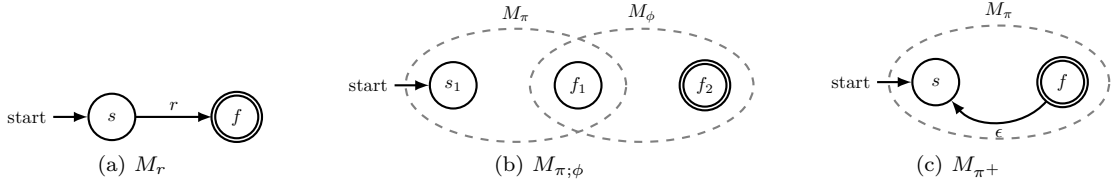
\begin{figure}[!ht]
  \subfloat[$M_r$]{
   \begin{tikzpicture}[->,>=latex,
           every state node/.style={draw,minimum width=20pt},thick,
           every node/.append style={scale=0.7, transform shape},
           ampersand replacement=\&]
        \node[initial,state] (q1) {$s$};
        \node[state,accepting,right=of q1] (f) {$f$};
        \draw[thick] (q1) to node[auto] {$r$} (f);
    \end{tikzpicture}
  }
  \hfill
  \subfloat[$M_{\pi \comp \phi}$]{
    \begin{tikzpicture}[->,>=latex,
           every state node/.style={draw,minimum width=20pt},thick,
           every node/.append style={scale=0.7, transform shape},
           ampersand replacement=\&]
        \draw[color=gray,dashed] (-0.275,0) ellipse [x radius=1.45cm, y radius=0.75cm];
        \node at (-0.275,0.95) {$M_{\pi}$};
        \draw[color=gray,dashed] (1.38,0) ellipse [x radius=1.47cm, y radius=0.75cm];
        \node at (1.38,0.95) {$M_{\phi}$};
        \matrix[column sep=1cm,row sep=1.25cm] at (0,0) {%
            \node[initial,state] (q1) {\entity{s_1}}; \&
            \node[state] (q2) {\entity{f_1}}; \&
            \node[state,accepting] (q3) {\entity{f_2}}; \\
        };
    \end{tikzpicture}}
   \hfill
   \subfloat[$M_{\pi^+}$]{
    \begin{tikzpicture}[->,>=latex,
           every state node/.style={draw,minimum width=20pt},thick,
           every node/.append style={scale=0.7, transform shape},
           ampersand replacement=\&]
        \draw[color=gray,dashed] (0.55,0) ellipse [x radius=1.45cm, y radius=0.75cm];
        \node at (0.55,0.95) {$M_{\pi}$};
        \matrix[column sep=1cm,row sep=1.25cm] at (0,0) {%
            \node[initial,state] (q1) {\entity{s}}; \&
            \node[state,accepting] (q2) {\entity{f}}; \\
        };
        \draw[thick] (q2) to [bend left=50] node[auto] {\rel{\epsilon}} (q1);
    \end{tikzpicture}}
    \caption{Schematic representations of NFAs for $r$, $\pi \comp \phi$ and $\pi^+$}\label{img:nfa-constructions-with-epsilon}
\end{figure}

The construction of the intersection NFA is simpler if the component NFAs do not contain empty transitions.
Accordingly, we modify the NFA for  $M_{\pi^+}$  to remove the empty transition.
Since every NFA representing a path condition has a single transition from the initial state, we may assume that we can write any path condition $\pi \ne \diamond$ in the form $r \comp \pi'$, where $r \in \RcoR$.
Hence, we may represent $\pi^+$ by the NFA shown in Figure~\ref{img:NFA:kleene-plus-without-epsilon:a}.
In the special case that $\pi = r$ for some $r \in \RcoR$, $\pi' = \diamond$ and the start and final states of $\pi'$ coincide, as shown in Figure~\ref{img:NFA:kleene-plus-without-epsilon:b}.

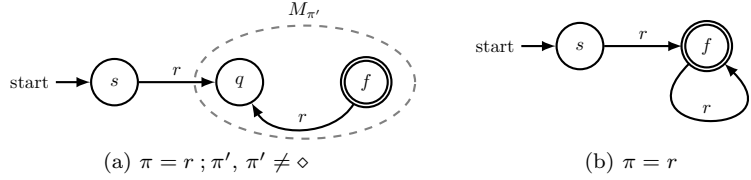
\begin{figure}[!ht]\centering
   \subfloat[$\pi = r \comp \pi'$, $\pi' \ne \diamond$]{
        \begin{tikzpicture}[->,>=latex,
               every state node/.style={draw,minimum width=20pt},thick,
               every node/.append style={scale=0.7, transform shape},
               ampersand replacement=\&]
        \draw[color=gray,dashed] (1.4,0) ellipse [x radius=1.45cm, y radius=0.75cm];
        \node at (1.4,0.95) {$M_{\pi'}$};
            \matrix[column sep=1cm,row sep=1.25cm] at (0,0) {%
                \node[initial,state] (v1) {\entity{s}}; \&
                \node[state] (v2) {\entity{q}}; \&
                \node[state,accepting] (v3) {\entity{f}}; \\
            };
            \draw[thick] (v1) to node[auto] {\rel{r}} (v2);
            \draw[thick] (v3) to [bend left=60] node[auto,swap] {\rel{r}} (v2);
        \end{tikzpicture}
        \label{img:NFA:kleene-plus-without-epsilon:a}
    } \quad
    \subfloat[$\pi = r$]{
        \begin{tikzpicture}[->,>=latex,
               every state node/.style={draw,minimum width=20pt},thick,
               every node/.append style={scale=0.7, transform shape},
               ampersand replacement=\&]
            \matrix[column sep=1cm,row sep=1.25cm] at (0,0) {%
                \node[initial,state] (v1) {\entity{s}}; \&
                \node[state,accepting] (v2) {\entity{f}}; \\
            };
            \draw[thick] (v1) to node[auto] {\rel{r}} (v2);
            \draw[thick] (v2.south west) to [loop below] node[auto] {$\rel{r}$} (v2.south east);
        \end{tikzpicture}
        \label{img:NFA:kleene-plus-without-epsilon:b}
    }
    \caption{$M_{\pi^+}$ without the empty transition}\label{img:NFA:kleene-plus-without-epsilon}
\end{figure}

\begin{example}
Consider the path condition
\[
 \overline{\left(\overline{r_1 \comp r_2^+}\right)^+ \comp (r_1 \comp r_3)^+}.
\]
Then, we may transform this into a simple path condition using the rules in Proposition~\ref{pro:simple-equivalences}.
\begin{align*}
 \overline{(\overline{r_1 \comp r_2^+})^+ \comp (r_1 \comp r_3)^+} %
   &= \overline{\left(r_1 \comp r_3\right)^+} \comp \overline{\left(\overline{r_1 \comp r_2^+}\right)^+} \\
   &= (\overline{r_1 \comp r_3})^+ \comp \overline{\left(\overline{r_1 \comp r_2^+}\right)}^+ \\
   &= (\overline{r_3} \comp \overline{r_1})^+ \comp (r_1 \comp r_2^+)^+
\end{align*}
The corresponding NFA is shown in Figure~\ref{fig:example-nfa-for-complexity}.
Note the number of states is $5$ and the number of transitions is $7$.
In Section~\ref{sec:requests:correctness_and_complexity} we establish the way in which the number of states and transitions vary with the structure of $\pi$.

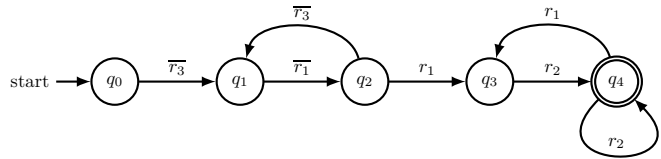
\begin{figure}[!ht]\centering
        \begin{tikzpicture}[->,>=latex,
               every state node/.style={draw,minimum width=20pt},thick,
               every node/.append style={scale=0.7, transform shape},
               ampersand replacement=\&]
            \matrix[column sep=1cm,row sep=1.25cm] at (0,0) {%
                \node[initial,state] (q0) {\entity{q_0}}; \&
                \node[state] (q1) {\entity{q_1}}; \&
                \node[state] (q2) {\entity{q_2}}; \&
                \node[state] (q3) {\entity{q_3}}; \&
                \node[state,accepting] (q4) {\entity{q_4}}; \\
            };
            \draw[thick] (q0) to node[auto] {$\overline{\rel{r_3}}$} (q1);
            \draw[thick] (q1) to node[auto] {$\overline{\rel{r_1}}$} (q2);
            \draw[thick] (q2) to [bend right=75] node[auto,swap] {$\overline{\rel{r_3}}$} (q1);
            \draw[thick] (q2) to node[auto] {\rel{r_1}} (q3);
            \draw[thick] (q3) to node[auto] {\rel{r_2}} (q4);
            \draw[thick] (q4) to [bend right=75] node[auto,swap] {\rel{r_1}} (q3);
            \draw[thick] (q4.south west) to [loop below] node[auto] {$\rel{r_2}$} (q4.south east);
        \end{tikzpicture}
    \caption{The NFA for $(\overline{r_3} \comp \overline{r_1})^+ \comp (r_1 \comp r_2^+)^+$} \label{fig:example-nfa-for-complexity}
\end{figure}
\end{example}

\subsubsection{Principal-matching}

The set of matched principals for a request is determined by identifying those principal-matching rules that are applicable to a given request $q = (s,o,a)$.
Recall that a system graph $G = (V,E)$ contains a set of nodes $V$ and a set of edges \mbox{$E \subseteq V \times V \times \RcoR$}, where $\RcoR$ is the set of relationship labels.
Given a request $q = (s,o,a)$ and the system graph $G = (V,E)$, we define the NFA $M_{q} = (V, \RcoR, E, s, \set{o})$.
Thus, every labelled edge in $G$ defines a transition and the start and final states are $s$ and $o$, respectively.

It is trivial to decide whether a request matches the $\all$ and $\none$ targets.
Hence, we only consider targets that are path conditions.
Informally, given a path condition $\pi$, a request $q$ and a system graph $G$, we wish to find a path in the directed graph $G$ (equivalently a word accepted by $M_q$) that is also a word accepted by $M_\pi$.
Thus, for a principal-matching rule $(\ppmc,\npmc,p)$ to be applicable to a request, where $\ppmc$ and $\npmc$ are path conditions, we require $\ppmc$ to be matched by the request and $\npmc$ to not be matched.
Therefore, we compute $L(M_{\ppmc}) \cap L(M_q)$ and $L(M_{\npmc}) \cap L(M_q)$; the former must be non-empty  and the latter must be empty.
We test these language properties by constructing two intersection automata, one from $M_{\ppmc}$ and $M_q$, the second from $M_{\npmc}$ and $M_q$.

\subsection{Complexity}\label{sec:requests:correctness_and_complexity}

The algorithms \textsf{ComputeAuthorizations} and \textsf{RequestEvaluation} do not involve significant computation.
The worst case time complexity of request evaluation is, therefore, dominated by the complexity of \textsf{ComputePrincipals}, which is dependent on two things:\footnote{The NFAs for path conditions contained in rules in the PMP can be pre-computed once and used, as required, to construct the intersection automata.}
\begin{itemize}
    \item the number of principal-matching rules to be evaluated; and
    \item the complexity of determining whether the intersection NFAs accept non-empty languages.
\end{itemize}
To evaluate the second of these factors, we define the \emph{length} $\ell(\pi)$ of a (simple) path condition $\pi$ to be:
    \begin{itemize}
        \item  $\ell(\pi) = 1$ if $\pi = r$ for some $r \in \RcoR$; 
        \item  $\ell(\pi \comp \pi') = \ell(\pi) + \ell(\pi')$;
        \item  $\ell(\pi^+) = \ell(\pi)$.
    \end{itemize}
In other words, $\ell(\pi)$ is simply the number of occurrences of elements in $\RcoR$ in $\pi$.
We now consider the size of the NFA $M_\pi$.

\begin{proposition}
Let $r \in \RcoR$, $\pi$ and $\phi$ be path conditions.  Then:
 \begin{itemize}
  \item $\card{Q_r} = 2$ and $\card{\delta_r} = 1$ for $r \in \RcoR$;
  \item $\card{Q_{\pi \comp \phi}} = \card{Q_{\pi}} + \card{Q_{\phi}} - 1$ and $\card{\delta_{\pi \comp \phi}} = \card{\delta_{\pi}} + \card{\delta_{\phi}}$;
  \item $\card{Q_{\pi^+}} = \card{Q_\pi}$ and $\card{\delta_{\pi^+}} = \card{\delta_{\pi}} + 1$.
 \end{itemize}
\end{proposition}

\begin{proof}
 The proof follows immediately by inspection of the NFAs in Figure~\ref{img:nfa-constructions-with-epsilon}.
\end{proof}

\begin{corollary}
 Let $\pi$ be a simple path condition and let $\vartheta(\pi)$ denote the number of occurrences of $+$ in $\pi$.
 Then for path condition $\pi$, $\card{Q_\pi} = \ell(\pi) + 1$ and $\card{\delta_\pi} = \ell(\pi) + \vartheta(\pi)$.
\end{corollary}

\begin{proof}
 The result may be proved by a straightforward induction on the structure of $\pi$.
 Clearly, the result for $\card{Q_\pi}$ holds for path condition $\pi = r$, $r \in \RcoR$.
 Now consider path condition $\pi \comp \phi$ and assume the result holds for $\pi$ and $\phi$.
 Then \[ \card{Q_{\pi \comp \phi}} = \card{Q_{\pi}} + \card{Q_{\phi}} - 1 = (\ell(\pi) + 1) + (\ell(\phi) + 1) - 1 = \ell(\pi \comp \phi) + 1, \] as required.
 Finally, consider path condition $\pi^+$ and assume the result holds for $\pi$.
 Then \[ \card{Q_{\pi^+}} = \card{Q_{\pi}} = \ell(\pi) + 1 = \ell(\pi^+) + 1. \]

 Similarly, the result for $\card{\delta_\pi}$ holds for path condition $\pi = r$.
 Now consider path condition $\pi \comp \phi$ and assume the result holds for $\pi$ and $\phi$.
 Then
 \[
  \card{\delta_{\pi \comp \phi}} = \card{\delta_\pi} + \card{\delta_{\phi}} %
				 = \ell(\pi) + \vartheta(\pi) + \ell(\phi) + \vartheta(\phi) %
				 = \ell(\pi \comp \phi) + \vartheta(\pi \comp \phi).
 \]
 Finally, consider $\pi^+$ and assume the result holds for $\pi$.
 Then \[ \card{\delta_{\pi^+}} = \card{\delta_\pi} + 1 = \ell(\pi) + \vartheta(\pi) + 1 = \ell(\pi^+) + \vartheta(\pi^+). \]
\end{proof}

The complexity of computing the intersection NFA for NFAs $M$ and $M'$ is determined by the size of the respective transition relations, since we compute a product automaton whose transition relation is determined by the transition relations of the component NFAs.
In the worst case, the size of transition relation $\delta \subseteq Q \times Q \times \Sigma$ is $O(\card{Q}^2 \cdot \card{\Sigma})$.
However, in the case of our path condition NFAs, we have $\card{\delta_\pi} = \ell(\pi) + \vartheta(\pi)$.
The size of the transition relation in $M_q$ is $O(|\RcoR| \cdot \card{V}^2)$.
Thus, the overall complexity of evaluating a path condition $\pi$ with respect to a request and system graph $G$ is $O((\ell(\pi) + \vartheta(\pi)) \cdot |\RcoR| \cdot \card{V}^2)$.
Each principal-matching rule contains at most two path conditions as targets.
And each principal-matching rule in policy $\rho$ must be evaluated.
Thus the overall complexity of evaluating a request is $O(\card{\rho} \cdot \vartheta(\rho) \cdot |\RcoR| \cdot \card{V}^2)$, where $\vartheta(\rho) = \max\set{\ell(\pi) + \vartheta(\pi) : \pi \in \rho}$.%
\footnote{In our previous work we had made use of a (modified) breadth-first search algorithm which resulted in a worst case time complexity of request evaluation of \mbox{$O((|\rho| \cdot \ell(\rho) \cdot |V|) + (|\rho| \cdot |\RcoR| \cdot |V|^2))$}~\cite{CramptonS14}.}

\subsection{Target-based Request Evaluation}\label{sec:requests:target-based}

The request evaluation process described above evaluates every principal-matching rule to identify those principals applicable to a request.
It subsequently determines whether those principals are authorized to perform the requested action on the object.
This process is rather inefficient, as one or more of the principals matched in this way may not appear in any authorization rules associated with the requested object.
Accordingly, we now briefly consider ways in which the request evaluation process could be improved.
The natural approach is to use ``target-based'' evaluation of requests, as exemplified by XACML~\cite{XACML3} and other target-based languages~\cite{BrunsH11,CramptonM12}.

We outline a simple target-based strategy, leaving further development for future work.
Given a request $(s,o,a)$, the only principals that can be relevant to evaluating the request are those that appear in rules of the form $(p,x,y,b)$, where $x \in \set{o,\tau(o),\star}$ and $y \in \set{a,\star}$.
The resulting list of principals can be used to limit the principal-matching rules that are evaluated.
In the worst case, of course, we still have to evaluate the entire set of principal-matching rules.%
\footnote{Target-based request evaluation as described is only applicable for list-oriented policies (not for graph-based policies). Nevertheless, it is straightforward to modify target-based evaluation to work in conjunction with policy graphs.  We omit the details, again deferring this for future work.}

\section{Typed Edges}\label{sec:extended_typed_edges}
The basic RPPM model, described in the previous sections, offers a general model of relationship-based access control which naturally supports contextual information.
It can be employed to describe structurally simple social network systems or online social networks~\cite{ChengPS12passat,ChengPS12dbsec,Fong11}.
However, the RPPM model can equally well describe far more complex systems representing individual computers, computer networks and organisations~\cite{CramptonS14}.
Whilst the basic RPPM model can cover a wide range of systems, the compute principals step of request evaluation may be computationally intensive for very large system graphs or systems employing principal-matching policies which contain a very large number of rules.
Additionally, it may be necessary for an implementation to support certain useful policy frameworks within the access control model.
For example, reputation and history-based access control (HBAC) systems rely on knowledge of previous actions to inform decisions~\cite{AbadiF03,EdjlaliAC99,KrukowNS08}.
More generally, workflow systems may use previous activity to enforce constraints such as separation of duty~\cite{GligorGF98,SimonZ97} and ``Chinese Walls''~\cite{BrewerN89}.

In order to broaden the applicability of the model and to improve the performance of request evaluation, the basic RPPM model can be extended to support \emph{typed edges}, where each edge in the system graph has a type.
\begin{itemize}
    \item \emph{Relationship edges} are the standard edge type and are labelled from the basic model's set of relationship labels $\RcoR$;
    \item \emph{Caching edges} enhance the performance of request evaluation processing and are labelled with a set of principals;
    \item \emph{Decision audit edges} record the decisions from previous authorization requests and are labelled with an indication of whether a requested action $a$ was authorized $\audita{a}$ or denied $\auditd{a}$ to a subject on an object;
    \item \emph{Interest audit edges} record a subject's active or blocked interest, \interesta  or \interestb respectively, in an entity.
\end{itemize}
Whilst relationship edges may be directed or undirected,  edges of the other types are always directed away from the relevant subject.

Relationship edges are the foundation of the system graph, representing relationships between entities in the system.
When a system is first described using the RPPM model, the system graph will only contain relationship edges.
As requests are made and evaluated, edges of the other types may be added to the system graph and can, therefore, impact future request evaluations.
Whilst the existence of decision and interest audit edges can alter the outcome of request evaluation, as discussed in Section~\ref{sec:extended_typed_edges:decision_audit_edges} and~\ref{sec:extended_typed_edges:interest_audit_edges}, the existence of caching edges may simply allow the first step of request evaluation to be bypassed, thus speeding up request evaluation without altering the final decision which results.

\subsection{Caching Edges}\label{sec:extended_typed_edges:caching_edges}
The first step of request evaluation, described in Section~\ref{sec:requests:path-conditions-to-principal-matching}, determines the set of matched principals which apply to a request.
The target-based optimisation for request evaluation, described in Section~\ref{sec:requests:target-based}, offers a potential reduction in the processing required during the compute principals step.
However, the benefit is limited to simply excluding some of the principal-matching rules during the evaluation of individual requests.
Which rules, if any, that are excluded varies with each request but the remainder must be processed as before.
However, note that the set of matched principals for a subject-object pair remains static until a change is made to the system graph or certain associated policy components.
This observation suggests a more widely applicable optimization.

Accordingly, we introduce the concept of \emph{caching edges} and make use of the relative stability of matched principals in order to reduce the processing required for future authorization requests.
In particular, when we evaluate a request $(s,o,a)$ that results in a set of matched principals $\semantics{\rho} \subseteq P$, we add an edge $(s,o,\semantics{\rho})$ to the system graph, directed from $s$ to $o$ and labelled with $\semantics{\rho}$; this edge identifies the matching principals relevant to future requests of the form $(s,o,a')$.
The processing of subsequent authorization requests of the form $(s,o,a')$ can, therefore, skip the computationally expensive step of computing the matched principals and instead use $\semantics{\rho}$ in conjunction with the authorization rules to evaluate requests of this form.
Recall that the action of a request is not part of the compute principals step of request evaluation.

\begin{example}
Recall that $u_1$ is the teaching assistant for course $c_2$ and is thus associated with the principal $p_2$ and authorized to read and grade $a_3$.
Suppose that $u_1$ makes the request $(u_1,a_3,\act{read})$.
Then this request will be authorized because $\semantics{\rho} = \set{p_2}$.
At this stage, we may therefore add an edge $(u_1,a_3,\set{p_2})$, thereby caching the outcome of the principal-matching phase of request evaluation, as illustrated in Figure~\ref{img:system-graph-fragment1}.
Then a subsequent request $(u_1,a_3,\act{grade})$ will only need to determine that $p_2$ is associated with the subject-object pair $(u_1,a_3)$ and can immediately evaluate the authorization rules (and thereby authorize the request).
\begin{figure}[!ht]\centering\setlength{\extrarowheight}{2pt}
         \begin{tikzpicture}
            [node distance=1.5cm and 1.6cm, on grid,
            caching/.style={color=blue!70,>=open diamond}, 
            audit-a/.style={color=teal, densely dashed}, 
            audit-d/.style={color=red, densely dotted}, 
            interest-a/.style={color=purple,>=\interestahead}, 
            interest-b/.style={color=brown,>=\interestdhead}, 
            every circle node/.style={draw,minimum width=20pt},thick,
            every node/.append style={scale=0.7, transform shape}]
            \begin{scope}[>=latex] 
                \node[circle] (u1) {\entity{u_1}};
                \node[circle,left=of u1] (c1) {\entity{c_1}};
                \node[circle,right=of u1] (c2) {\entity{c_2}};
                \node[circle,left=of c1] (a1) {\entity{a_1}};
                \node[circle,right=of c2] (a3) {\entity{a_3}};
                \node[circle,below=of c1] (a2) {\entity{a_2}};
                \node[circle,above=of c1] (u2) {\entity{u_2}};
                \draw[thick,->] (u1) to node[swap,auto] {$r_1$} (c1);
                \draw[thick,->] (u2) to node[auto] {$r_5$} (c1);
                \draw[thick,->] (u1) to node[auto] {$r_3$} (c2);
                \draw[thick,->] (a1) to node[auto] {$r_2$} (c1);
                \draw[thick,->] (a3) to node[swap,auto] {$r_2$} (c2);
                \draw[thick,->] (a2) to node[auto] {$r_2$} (c1);
                \draw[thick,->] (u1) to node[auto] {$r_4$} (a2);
                \draw[caching,thick,->] (u1) to[bend left] node[auto] {$\set{p_2}$} (a3);
              \end{scope}
      \end{tikzpicture}
      \caption{Adding a caching edge}\label{img:system-graph-fragment1}
\end{figure}
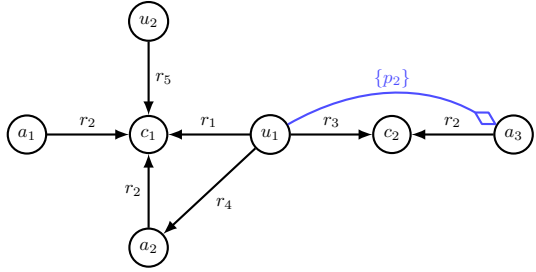
\end{example}

\subsubsection{Cache Management}\label{sec:extended_typed_edges:caching_edges:management}
Whilst the performance improvement offered by caching edges may be significant, there is a need to carefully manage any implementation of caching edges to prevent this improvement being countered by an indiscriminate increase in the number of edges in the system graph.
In the worst case, the number of caching edges directed out of a node is $O(|V|)$, where $V$ is the set of nodes in the system graph.
However, there are strategies that can be used to both prevent the system graph realizing the worst case and to reduce the impact of large numbers of caching edges.
To maintain an acceptable number of caching edges, we could, for example, use some form of cache purging~\cite{CramptonS14_STM}.
Further experimental work is required to determine how best to make use of caching edges.

\subsubsection{Preemptive Caching}\label{sec:extended_typed_edges:caching_edges:preemptive}
Any optimisation provided by the caching of matched principals relies upon the existence of a caching edge in order to reduce the authorization request processing.
The first request between a subject and object must, therefore, be processed normally in order to determine the set of matched principals which will label the caching edge.
If this initial evaluation were only performed when an authorization request were submitted, then the benefit of caching edges would be limited to repeated subject-object interactions alone.

However, many authorization systems will experience periods of time when no authorization requests are being evaluated.
The nature of many computing tasks is such that authorization is required sporadically amongst longer periods of computation by clients of the authorization system and idle time for the authorization system itself.
These periods of reduced load on the authorization system can be employed for the purpose of \emph{preemptive caching}~\cite{CramptonS14_STM}.
Again, further experimental work is required to evaluate how best to perform preemptive caching.

\subsection{Decision Audit Edges}\label{sec:extended_typed_edges:decision_audit_edges}
The basic form of the RPPM model is ``memoryless'' with respect to request evaluations.
The introduction of decision audit edges allows the system to record whether previously requested actions were authorized or denied.
Both authorized and denied decision audit edges are inserted, automatically, into the system graph after request evaluation completes.
If such an edge does not already exist, a decision audit edge is added between the subject and object of the evaluated request, indicating its result.%
\footnote{In some situations there may be merit in recording every occurrence of an action (by a subject on an object) being authorized or denied. Modifying the decision audit edge's label to include a count would be a simple way of achieving this if it were required.}

These edges can be utilised in a variety of ways depending on the requirements of the authorization system.
At their simplest, they provide a record which may be used as input for auditing or other processing outside of the authorization system.%
\footnote{In this form they may, for example, be used to identify potentially malicious actors.
A node with a large number of, or a sudden increase in, denied decision audit edges may be considered worthy of further investigation.}
Within the authorization system, decision audit edges can be used to match part of a path condition, thus enabling authorization decisions to be made based on historical evidence.

The caching edges we introduced in the previous section arise from the principal-matching part of the request evaluation process.
An audit edge arises from the second phase of the evaluation process.
We extend the set of relationship labels by defining the relationships  $\audita{a}$ and $\auditd{a}$ for each action $a$.
\begin{itemize}
 \item If the decision for request $(s,o,a)$ is allow, then we add the edge $(u,o,\audita{a})$ to the system graph.
 \item Conversely, if the decision for request $(u,o,a)$ is deny, then we add the edge $(u,o,\auditd{a})$ to the system graph.
\end{itemize}

Principal-matching rules can be created to make direct use of decision audit edges.
Some obvious examples include:
 \begin{itemize}
  \item the principal-matching rule $(\all, \audita{a}, p)$ can be used to match the principal $p$ to any request where the subject has not previously performed the action $a$ on the object;
  \item the rule $(\all, \auditd{a}, p)$ requires that the subject has never been denied action $a$ on the object;
  \item the rule $(\audita{a}, \none, p)$ requires that the subject must have previously had a request to perform action $a$ on the object approved.
 \end{itemize}

\begin{example}\label{ex:higher-education-2}
Returning to our higher education example, suppose that we have a student $u_3$ who is enrolled on course $c_2$ and is the author of coursework $a_3$.
Then $u_1$, the teaching assistant for the course will, at some point, grade the coursework $a_3$.
At this point, $u_3$ should not be able to modify $a_3$.
We could enforce this requirement by modifying the principal-matching rule $(r_4,\none,p_1)$---which assigns any user who is the owner of a piece of coursework to the author principal---to $(r_4,r_1 \comp \overline{r_3} \comp \audita{grade},p_1)$.
This rule includes a second path condition $r_1 \comp \overline{r_3} \comp \audita{grade}$, one that must not be matched if the principal $p_1$ is to apply to a request.
This path condition traces a path from enrolled student to course to teaching assistant to (graded) coursework.
Figure~\ref{img:he-system-graph-fragment2} illustrates the system graph once the teaching assistant has graded $a_3$; we represent allow audit edges using dashed lines.
Note that there is a path from $u_3$ to $a_3$ matching the prohibited path condition.

Of course, in practice $u_3$ will still wish to read $a_3$, so we might wish to specify a separate rule, rather than modify the existing rule.
This separate rule could have the form $(r_1 \comp \overline{r_3} \comp \audita{grade}, \none, p)$ and then we specify an additional authorization rule $(p,\star,\act{write},0)$ which explicitly denies principal $p$ write access to any object.

\begin{figure}[!ht]\centering
         \begin{tikzpicture}
            [node distance=1.5cm and 1.6cm, on grid,
            caching/.style={color=blue!70,>=open diamond}, 
            audit-a/.style={color=teal, densely dashed}, 
            audit-d/.style={color=red, densely dotted}, 
            interest-a/.style={color=purple,>=\interestahead}, 
            interest-b/.style={color=brown,>=\interestdhead}, 
            every circle node/.style={draw,minimum width=20pt},thick,
            every node/.append style={scale=0.7, transform shape}]
            \begin{scope}[>=latex] 
                \node[circle] (u1) {\entity{u_1}};
                \node[circle,left=of u1] (c1) {\entity{c_1}};
                \node[circle,right=of u1] (c2) {\entity{c_2}};
                \node[circle,left=of c1] (a1) {\entity{a_1}};
                \node[circle,right=of c2] (a3) {\entity{a_3}};
                \node[circle,below=of c1] (a2) {\entity{a_2}};
                \node[circle,above=of c1] (u2) {\entity{u_2}};
                \node[circle,above=of c2] (u3) {\entity{u_3}};
                \draw[thick,->] (u1) to node[swap,auto] {$r_1$} (c1);
                \draw[thick,->] (u2) to node[auto] {$r_5$} (c1);
                \draw[thick,->] (u1) to node[auto] {$r_3$} (c2);
                \draw[thick,->] (a1) to node[auto] {$r_2$} (c1);
                \draw[thick,->] (a3) to node[swap,auto] {$r_2$} (c2);
                \draw[thick,->] (a2) to node[auto] {$r_2$} (c1);
                \draw[thick,->] (u1) to node[auto] {$r_4$} (a2);
                \draw[thick,->] (u3) to node[auto] {$r_1$} (c2);
                \draw[thick,->] (u3) to node[auto] {$r_4$} (a3);
                \draw[audit-a,->] (u1) to[bend right] node[auto,swap] {$\audita{grade}$} (a3);
              \end{scope}
      \end{tikzpicture}
     \caption{Using audit edges to enforce complex policy requirements}\label{img:he-system-graph-fragment2}
\end{figure}
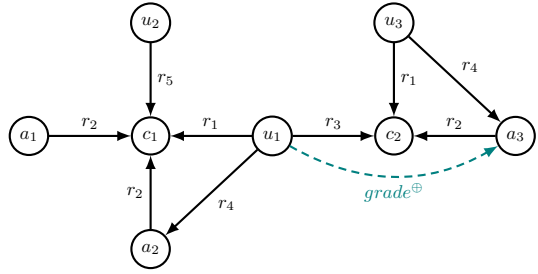
\end{example}

\subsubsection{Separation of Duty}\label{sec:extended_typed_edges:decision_audit_edges:separation_of_duty}
Whilst decision audit edges can be used in an ad hoc manner to enforce application-specific constraints, we can also use them to enforce separation-of-duty in a systematic way.
Separation of duty requires that certain combinations of actions are performed by a number of distinct individuals so as to reduce the likelihood of abuse of a system.
In its simplest form, separation of duty constraints require two individuals to each perform one of a pair of distinct actions so that a single individual cannot abuse the system.
A common application environment for such constraints is that of a finance system, where, for example, the individual authorized to add new suppliers should not be the same individual who is authorized to approve the payment of invoices to suppliers.
If a single individual were able to perform both of these actions they could set themselves up as a supplier within the finance system and then approve for payment any invoices they submitted as that supplier.

We define a mechanism here through which $n$ (distinct) users are required to perform $n$ actions associated with an object.
Let us consider the system graph $G_2$ (see Figure~\ref{img:system-graph-fragment2:a}), a set of actions $A = \set{a_1,\dots,a_m}$, $m \geqslant n$, and the following policies
\begin{align*}
    \rho &= \set{(r,\none,p)} \\
    \varrho &= \set{(p,o,\star,1)} \\
    \crs &= \crso{DenyOverride}
\end{align*}
With these policies (whether we use audit edges or not), if individual $u_1$ makes the request $q_1 = (u_1,o,a_1)$ this will be authorized by matching principal $p$, as will subsequent requests $q_2 = (u_1,o,a_2)$ and $q_3 = (u_1,o,a_3)$.
A similar result would have occurred if these requests had been submitted with $u_2$ or $u_3$ as the subject.

\begin{figure}[!ht]\centering
    \subfloat[system graph fragment]{
        \begin{tikzpicture}
            [node distance=1cm and 1.5cm,
            caching/.style={color=blue!70,>=open diamond}, 
            audit-a/.style={color=teal, densely dashed}, 
            audit-d/.style={color=red, densely dotted}, 
            interest-a/.style={color=purple,>=\interestahead}, 
            interest-b/.style={color=brown,>=\interestdhead}, 
              every circle node/.style={draw,minimum width=20pt},thick,
              every node/.append style={scale=0.7, transform shape}]
              \begin{scope}[>=latex] 
            	  \node[circle] (u1) {\entity{u_1}};
            	  \node[circle,right=of u1] (o) {\entity{o}};
            	  \node[circle,below=of o] (u2) {\entity{u_2}};
            	  \node[circle,right=of o] (u3) {\entity{u_3}};
            	  \draw[thick,->] (u1) to node[auto] {\rel{r}} (o);
            	  \draw[thick,->] (u2) to node[auto] {\rel{r}} (o);
            	  \draw[thick,->] (u3) to node[swap,auto] {\rel{r}} (o);
              \end{scope}
        \end{tikzpicture}
        \label{img:system-graph-fragment2:a}
    }\qquad
     \subfloat[after request $q_6 = (u_2,o,a_3)$]{
        \begin{tikzpicture}
            [node distance=1cm and 2cm,
            caching/.style={color=blue!70,>=open diamond}, 
            audit-a/.style={color=teal, densely dashed}, 
            audit-d/.style={color=red, densely dotted}, 
            interest-a/.style={color=purple,>=\interestahead}, 
            interest-b/.style={color=brown,>=\interestdhead}, 
              every circle node/.style={draw,minimum width=20pt},thick,
              every node/.append style={scale=0.7, transform shape}]
              \begin{scope}[>=latex] 
            	  \node[circle] (u1) {\entity{u_1}};
            	  \node[circle,right=of u1] (o) {\entity{o}};
            	  \node[circle,below=of o] (u2) {\entity{u_2}};
            	  \node[circle,right=of o] (u3) {\entity{u_3}};
            	  \draw[thick,->] (u1) to node[auto] {\rel{r}} (o);
            	  \draw[thick,->] (u2) to node[auto] {\rel{r}} (o);
            	  \draw[thick,->] (u3) to node[swap,auto] {\rel{r}} (o);
                  \draw[audit-a,thick,->] (u1) to [bend left=35] node[auto,color=black] {\rel{\audita{a_1}}} (o);
                  \draw[audit-d,thick,->] (u1) to [bend right=25] node[swap,auto,color=black] {\rel{\auditd{a_2}}} (o);
                  \draw[audit-d,thick,->] (u1) to [bend right=55] node[swap,auto,color=black] {\rel{\auditd{a_3}}} (o);
                  \draw[audit-a,thick,->] (u2) to [bend right] node[swap,auto,color=black] {\rel{\audita{a_3}}} (o);
                  \draw[audit-a,thick,->] (u3) to [bend right=35] node[swap,auto,color=black] {\rel{\audita{a_2}}} (o);
                  \draw[audit-d,thick,->] (u3) to [bend left=35] node[auto,color=black] {\rel{\auditd{a_3}}} (o);
              \end{scope}
        \end{tikzpicture}
        \label{img:system-graph-fragment2:d}
    }
    \caption{Adding decision audit edges}\label{img:system-graph-fragment2}
\end{figure}
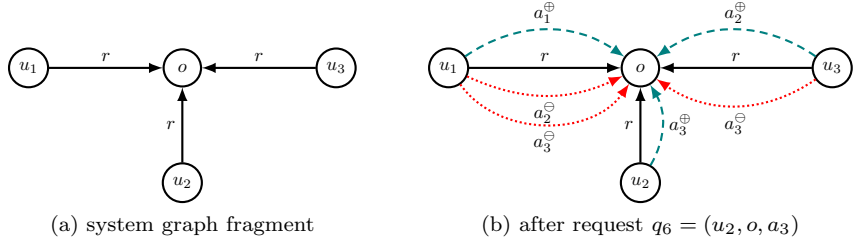

Now suppose we wish to restrict each user to a single interaction with $o$.
Then we define the policies
\begin{align*}
 \rho' &= \set{(\audita{a_i},\none,p_i) : 1 \leqslant i \leqslant n} \cup \rho \\
 \varrho' &= \set{(p_i,o,a_j,0) : 1 \leqslant i \leqslant n, j \ne i} \cup \varrho \\
\end{align*}
Then, assuming there are no audit edges in the system graph, request $(u_1,o,a_1)$ matches the rule in $\rho$, as before, and the request is authorized by the rule in $\varrho$.
However, assuming the audit edge $(u_1,o,\audita{a_1})$ is now added to the system graph, a subsequent request $(u_1,o,a_2)$ will match the new rule $(\audita{a_1},\none,p_1)$, leading to a deny decision (because of the new rule $(p_1,o,a_2,0)$).
At this point a deny audit edge will be added to the system graph.
Similarly, request $(u_1,o,a_3)$ will be denied, and any attempt by $u_2$ or $u_3$ to perform two different actions will result in at most one allow decision.
The case $n=3$ is illustrated in Figure~\ref{img:system-graph-fragment2:d}, where each of the three users is permitted to perform one of the three actions.
More formally, we have the following result.

\begin{proposition}\label{prop:separation-of-duty}
    Given an RPPM separation of duty policy, as described above, for any user $u$ the request $(u,o,a)$ is allowed if the request is authorized by $\rho'$ and $(\varrho',\chi)$ and no request of the form $(u,o,a')$ has been previously authorized where $a' \neq a$ and $a, a' \in \set{a_1,\dots,a_n}$. The request is denied otherwise.
\end{proposition}

\begin{proof}
    The proof proceeds by induction on the number of evaluated requests.
    Consider the (base) case when no requests have yet been made.
    A request $(u,o,a)$ where $a \in \set{a_1, \dots, a_n}$ will not match $(\audita{a_1},\none,p_i)$ for any $i$, $1 \leqslant i \leqslant n$, as no decision audit edges currently exist in the system graph.
    Thus request $(u,o,a)$ will be authorized if it is authorized by $\rho$ and $(\varrho, \chi)$ (and hence will be authorized by $\rho'$ and $(\varrho',\chi)$).

    Now suppose the result holds for all sequences of $m$ requests and consider the request $(u,o,a)$ where $a \in \set{a_1, \dots, a_n}$.
    \begin{itemize}
        \item If $u$ has previously performed a constrained action $a_i$, $1 \leqslant i \leqslant n$, then the request will satisfy principal-matching rule $(\audita{a_i}, \none, p_i)$.
	
	      Now, if $a_i = a$, there is no authorization rule of the form $(p_i,o,a_i,0)$ and the request will, therefore, be authorized if and only if it is authorized by $\rho$ and $(\varrho, \chi)$.
	
	      Conversely, if $a_i \ne a$, then $a = a_j$, for some $j \ne i$, and the authorization rule $(p_i, o, a, 0)$, together with the \crso{DenyOverride} CRS will cause the request to be denied.
        \item If user $u$ has not previously performed a constrained action then the request will not match any of the principal-matching rules that were added to create $\rho'$.
	      Thus the request will only be authorized if it is authorized by $\rho$ and $(\varrho, \chi)$.
    \end{itemize}
\end{proof}

Whilst we have proposed a mechanism for separation of duty employing the $\crso{DenyOverride}$ conflict resolution strategy, it would be equally possible to employ list-oriented policies (as described in Section~\ref{sec:model:policy-extensions:list-oriented-policies}) assuming the added constraint rules are inserted at the start of the principal-matching policy.

\subsection{Interest Audit Edges}\label{sec:extended_typed_edges:interest_audit_edges}
In a similar way to decision audit edges, interest audit edges record information related to previous requests which can be utilised to make future decisions.
However, whilst the decision audit edges record the direct result of a previous request evaluation, interest audit edges record the higher level notion of ``interest'' associated with those requests.
A subject who requests to perform an action on an object can be considered to be showing an interest in that object (or an entity which that object is related to).
An authorization system may be configured to use the record of this interest to determine whether a future request on that, or another, object should be approved or denied.
The primary use case for interest audit edges is the support of Chinese Wall policies.

\subsubsection{Chinese Wall}\label{sec:extended_typed_edges:interest_audit_edges:chinese_wall}
The Chinese Wall principle may be used to control access to information in order to prevent any conflicts of interest arising.
The standard use case concerns a consultancy that provides services to multiple clients, some of which are competitors.
It is important that a consultant does not access documents of company $c$ if she has previously accessed documents of a competitor of $c$.

To support the Chinese Wall policy, data is classified using conflict of interest classes to indicate groups of competitor entities~\cite{BrewerN89}.
Requests to access a company's resources within a conflict of interest class will only be authorized if no previous request was authorized accessing resources from another company in that conflict of interest class.

Let us suppose that for a given conflict of interest class $c$ and every user $u$, we have $G,u,c \models \pi_1$, and for every resource $o$, we have $G,o,c \models \pi_2$.
Then for this conflict of interest class we have $G,u,o \models \pi_1 \comp \overline{\pi_2}$ for every user and resource.
This arrangement is depicted, conceptually, in Figure~\ref{img:chinesewallconcept:b}.%
\footnote{Figure~\ref{img:chinesewallconcept} does not show a system graph; it shows high-level representations of the ``shape'' of a system graph.}
We introduce a logical entity into the system graph to represent a conflict of interest class and the relationship $m$, where an edge $(c,i,m)$ indicates company $c$ is a member of conflict of interest class $i$.
(We assume here that membership of conflict of interest classes is determined when the system graph is initially populated and remains fixed through the lifetime of the system.)

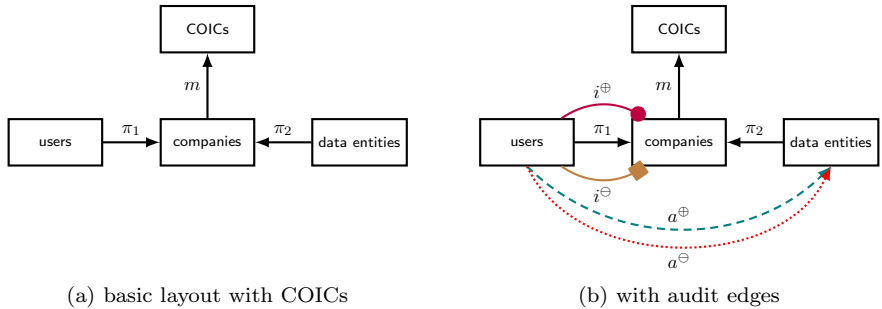
\begin{figure}[!ht]\centering
     \subfloat[basic layout with COICs]{
        \begin{tikzpicture}
            [node distance=1.5cm and 2cm, on grid,
            caching/.style={color=blue!70,>=open diamond}, 
            audit-a/.style={color=teal, densely dashed}, 
            audit-d/.style={color=red, densely dotted}, 
            interest-a/.style={color=purple,>=\interestahead}, 
            interest-b/.style={color=brown,>=\interestdhead}, 
            entity/.style={rectangle,draw,minimum width=50pt,minimum height=25pt},thick,
            every node/.append style={scale=0.7, transform shape}]
            \begin{scope}[>=latex] 
                \node[entity] (u) {\textsf{\footnotesize users}};
                \node[entity,right=of u] (c) {\textsf{\footnotesize companies}};
                \node[entity,right=of c] (d) {\textsf{\footnotesize data entities}};
                \node[entity,above=of c] (i) {\textsf{\footnotesize COICs}};
                \draw[thick,->] (u) to node[auto] {\rel{\pi_1}} (c);
                \draw[thick,->] (d) to node[swap,auto] {\rel{\pi_2}} (c);
                \draw[thick,->] (c) to node[auto] {\rel{m}} (i);
                \draw[audit-a,thick,->,opacity=0] (u.south) to [bend right=45] node[swap,auto,color=black] {\rel{\audita{a}}} (d.south);
                \draw[audit-d,thick,->,opacity=0] (u.south) to [bend right=65] node[swap,auto,color=black] {\rel{\auditd{a}}} (d.south);
            \end{scope}
        \end{tikzpicture}
        \label{img:chinesewallconcept:b}
    }\qquad
     \subfloat[with audit edges]{
        \begin{tikzpicture}
            [node distance=1.5cm and 2cm, on grid,
            caching/.style={color=blue!70,>=open diamond}, 
            audit-a/.style={color=teal, densely dashed}, 
            audit-d/.style={color=red, densely dotted}, 
            interest-a/.style={color=purple,>=\interestahead}, 
            interest-b/.style={color=brown,>=\interestdhead}, 
            entity/.style={rectangle,draw,minimum width=50pt,minimum height=25pt},thick,
            every node/.append style={scale=0.7, transform shape}]
            \begin{scope}[>=latex] 
                \node[entity] (u) {\textsf{\footnotesize users}};
                \node[entity,right=of u] (c) {\textsf{\footnotesize companies}};
                \node[entity,right=of c] (d) {\textsf{\footnotesize data entities}};
                \node[entity,above=of c] (i) {\textsf{\footnotesize COICs}};
                \draw[thick,->] (u) to node[auto] {\rel{\pi_1}} (c);
                \draw[thick,->] (d) to node[swap,auto] {\rel{\pi_2}} (c);
                \draw[thick,->] (c) to node[auto] {\rel{m}} (i);
                \draw[interest-a,thick,->] (u) to [bend left=35] node[auto,color=black] {\rel{\interesta}} (c);
                \draw[interest-b,thick,->] (u) to [bend right=35] node[swap,auto,color=black] {\rel{\interestb}} (c);
                \draw[audit-a,thick,->] (u.south) to [bend right=45] node[auto,color=black] {\rel{\audita{a}}} (d.south);
                \draw[audit-d,thick,->] (u.south) to [bend right=65] node[swap,auto,color=black] {\rel{\auditd{a}}} (d.south);
            \end{scope}
      \end{tikzpicture}
        \label{img:chinesewallconcept:c}
    }
    \caption{Chinese Wall generalisation}\label{img:chinesewallconcept}
\end{figure}

We now introduce interest audit edges into the system graph which are added between users and companies (see Figure~\ref{img:chinesewallconcept:c}).
Active interest audit edges are labelled with \interesta, whilst blocked interest audit edges are labelled with \interestb.
We, therefore, extend the set of relationships $\RcoR$ to include the set $\setinterest$, thus allowing the system graph to support these new edges.%
\footnote{Whilst we do not rely upon decision audit edges to enforce Chinese Wall policies they, equally, do not interfere with interest audit edges.
Our discussion of Chinese Wall will consider an authorization system which is also supporting decision audit edges so as to provide a more complete picture of an extended RPPM model.}
Therefore, when users are authorized (or denied) access to particular data entities, authorized (or denied) decision audit edges will result for these requests as shown in Figure~\ref{img:chinesewallconcept:c}.

Given a system graph $G = (V,E)$ such that $G,u,c \models \pi_1$ and $G,o,c \models \pi_2$ for all users $u$, all objects $o$ and all companies $c$ with membership of a given conflict-of-interest class $i$, the principal-matching rule $(\pi_1 \comp \overline{\pi_2},\none,p)$ ensures that every request of the form $(u,o,\act{read})$ is matched to principal $p$.
Hence the authorization rule $(p,\star,\act{read},1)$ authorizes every $u$ to read every $o$.
Now suppose that we wish to extend this basic policy and enforce a Chinese Wall policy, which requires that if a user $u$ reads a document belonging to company $c$ where $(c, i, m) \in E$ then $u$ must not read any document belonging to $c'$, where $c' \ne c$ and $(c', i, m) \in E$.
Then we redefine the principal-matching rule to be $(\pi_1 \comp \overline{\pi_2}, \interestb \comp \overline{\pi_2},p)$.\footnote{If our base principal-matching rule employed a prohibited target already then we could still enforce the Chinese Wall policy by inserting an additional principal-matching rule of the form $(\interestb \comp \overline{\pi_2},\none,p_{block})$ where the principal $p_{block}$ is denied all actions on all objects~\cite{CramptonS14_STM}.}
Consider an initial request $(u,o,\act{read})$, where $o$ is a document owned by $c$ a member of conflict-of-interest class $i$.
If principal $p$ is matched, that is
\[
 G,u,o \models \pi_1 \comp \overline{\pi_2}\quad \text{and}\quad G,u,o \not\models \interestb \comp \overline{\pi_2},
\]
then the request is authorized (since the principal $p$ is matched) and the following edges are added to $G$:
\begin{itemize}
 \item $(u,c,\interesta)$;
 \item $(u,c',\interestb)$ for all $c' \ne c$ where $(c,i,m) \in E$ and $(c',i,m) \in E$; and
 \item $(u,o,\audita{read})$.
\end{itemize}
Consider a subsequent request $(u,o',\act{read})$, where $o'$ is owned by $c' \ne c$ and $c'$ belongs to the same conflict-of-interest class as $c$.
Then $G,u,o' \models \interestb \comp \overline{\pi_2}$ and principal $p$ is no longer matched and the request will be denied (assuming we deny by default).

This is illustrated in Figure~\ref{img:system-graph-fragment4}, where $\pi_1 = w \comp s$ and $\pi_2 = d$.
A member of staff $u_1$ works for a consultancy firm $e_1$ that acts on behalf of clients ($c_1$, $c_2$ and $c_3$) and stores data about the commercial interests of those clients in the form of files ($f_1$, $f_2$, $f_3$ and $f_4$).
We denote edges of the form $(u,c,\interesta)$ with a filled circle head and those of the form $(u,c,\interestb)$ with a filled square head.
The figure illustrates the system graph before and after requests $(u_1,f_1,\act{read})$ and $(u_1,f_4,\act{read})$ have been authorized.
This results in additional edges in the system graph, notably $(u_1,c_2,\interestb)$, which means that $G,u_1,f_2 \models \interestb \comp \overline{\pi_2}$ and the request $(u_1,f_2,\act{read})$ would be denied (since $p$ would not be matched).
Note that request $(u_1,f_3,\act{read})$ would be permitted because $G,u_1,f_3 \not\models \interestb \comp \overline{\pi_2}$.
The audit and interest edges added through evaluation of these two further requests are also shown in Figure~\ref{img:system-graph-fragment4:c}.

Our discussion has so far considered the case where a single path of relationships exists between users and companies and between objects and companies; in reality there may be multiple alternative paths and the basic layout can be adjusted to enable this.\footnote{The key components of the basic layout are the existence of a subset of system graph entities $C \subset V$ (in our example companies) connecting users to objects, the fact that each $c \in C$ is a member of at most one conflict of interest classes $i$, and the fact that $C$ is the range for interest audit edges.}
To support such multi-path scenarios, the approach described above can be generalised as follows.
Let $\pi_1, \dots, \pi_n$ where $n \geqslant 1$ be paths between users and companies, and let $\pi'_1, \dots, \pi'_m$ where $m \geqslant 1$ be paths between objects and companies.
The set of principal-matching rules required to authorise users to access objects (through all possible combinations of these paths) is \[
    \set{(\pi_i \comp \overline{\pi'_j},\; \none,\; p) : 1 \leqslant i \leqslant n, 1 \leqslant j \leqslant m}.
\]
In order to support Chinese Wall policies, these rules are modified to
\[
    \set{(\pi_i \comp \overline{\pi'_j},\; \interestb \comp \overline{\pi'_j},\; p) : 1 \leqslant i \leqslant n, 1 \leqslant j \leqslant m}.
\]

\begin{figure}[!ht]\centering
    \subfloat[system graph fragment]{
        \begin{tikzpicture}
            [node distance=1.5cm and 1.6cm, on grid,
            caching/.style={color=blue!70,>=open diamond}, 
            audit-a/.style={color=teal, densely dashed}, 
            audit-d/.style={color=red, densely dotted}, 
            interest-a/.style={color=purple,>=\interestahead}, 
            interest-b/.style={color=brown,>=\interestdhead}, 
            every circle node/.style={draw,minimum width=20pt},thick,
            every node/.append style={scale=0.7, transform shape}]
            \begin{scope}[>=latex] 
                \node[circle] (u1) {\entity{u_1}};
                \node[circle,below=of u1] (e1) {\entity{e_1}};
                \node[circle,left=of e1] (c1) {\entity{c_1}};
                \node[circle,below=of e1] (c2) {\entity{c_2}};
                \node[circle,right=of e1] (c3) {\entity{c_3}};
                \node[circle,above=of c1] (f1) {\entity{f_1}};
                \node[circle,left=of c1] (f4) {\entity{f_4}};
                \node[circle,above=of c3] (f3) {\entity{f_3}};
                \node[circle,below=of c2] (f2) {\entity{f_2}};
                \node[circle,left=of c2] (i1) {\entity{i_1}};
                \node[circle,right=of c2] (i2) {\entity{i_2}};
                \draw[thick,->] (u1) to node[swap,auto] {\rel{w}} (e1);
                \draw[thick,->] (e1) to node[swap,auto] {\rel{s}} (c1);
                \draw[thick,->] (e1) to node[swap,auto] {\rel{s}} (c2);
                \draw[thick,->] (e1) to node[auto] {\rel{s}} (c3);
                \draw[thick,->] (f1) to node[swap,auto] {\rel{d}} (c1);
                \draw[thick,->] (f2) to node[auto] {\rel{d}} (c2);
                \draw[thick,->] (f3) to node[swap,auto] {\rel{d}} (c3);
                \draw[thick,->] (f4) to node[auto] {\rel{d}} (c1);
                \draw[thick,->] (c1) to node[swap,auto] {\rel{m}} (i1);
                \draw[thick,->] (c2) to node[swap,auto] {\rel{m}} (i1);
                \draw[thick,->] (c3) to node[swap,auto] {\rel{m}} (i2);
            \end{scope}
        \end{tikzpicture}
        \label{img:system-graph-fragment4:a}
    }\qquad
     \subfloat[after four requests]{
        \begin{tikzpicture}
            [node distance=1.5cm and 1.6cm, on grid,
            caching/.style={color=blue!70,>=open diamond}, 
            audit-a/.style={color=teal, densely dashed}, 
            audit-d/.style={color=red, densely dotted}, 
            interest-a/.style={color=purple,>=\interestahead}, 
            interest-b/.style={color=brown,>=\interestdhead}, 
            every circle node/.style={draw,minimum width=20pt},thick,
            every node/.append style={scale=0.7, transform shape}]
            \begin{scope}[>=latex] 
                \node[circle] (u1) {\entity{u_1}};
                \node[circle,below=of u1] (e1) {\entity{e_1}};
                \node[circle,left=of e1] (c1) {\entity{c_1}};
                \node[circle,below=of e1] (c2) {\entity{c_2}};
                \node[circle,right=of e1] (c3) {\entity{c_3}};
                \node[circle,above=of c1] (f1) {\entity{f_1}};
                \node[circle,left=of c1] (f4) {\entity{f_4}};
                \node[circle,above=of c3] (f3) {\entity{f_3}};
                \node[circle,below=of c2] (f2) {\entity{f_2}};
                \node[circle,left=of c2] (i1) {\entity{i_1}};
                \node[circle,right=of c2] (i2) {\entity{i_2}};
                \draw[thick,->] (u1) to node[swap,auto,near end] {\rel{w}} (e1);
                \draw[thick,->] (e1) to node[swap,auto] {\rel{s}} (c1);
                \draw[thick,->] (e1) to node[swap,auto] {\rel{s}} (c2);
                \draw[thick,->] (e1) to node[auto] {\rel{s}} (c3);
                \draw[thick,->] (f1) to node[swap,auto,pos=0.35] {\rel{d}} (c1);
                \draw[thick,->] (f2) to node[auto] {\rel{d}} (c2);
                \draw[thick,->] (f3) to node[swap,auto] {\rel{d}} (c3);
                \draw[thick,->] (f4) to node[auto] {\rel{d}} (c1);
                \draw[thick,->] (c1) to node[swap,auto] {\rel{m}} (i1);
                \draw[thick,->] (c2) to node[swap,auto] {\rel{m}} (i1);
                \draw[thick,->] (c3) to node[swap,auto] {\rel{m}} (i2);   	
                \draw[audit-a,thick,->] (u1) to node[swap,auto,color=black] {\rel{\audita{\act{read}}}} (f1);
                \draw[interest-a,thick,->] (u1) to node[color=black,fill=white,inner sep=1pt] {\rel{\interesta}} (c1);
                \draw[interest-b,thick,->] (u1) to [bend right] node[pos=0.7,color=black,,fill=white,inner sep=1pt] {\rel{\interestb}} (c2);
                \draw[audit-d,thick,->] (u1) to [bend left] node[auto,near end,color=black] {\rel{\auditd{\act{read}}}} (f2);
                \draw[audit-a,thick,->] (u1) to node[auto,color=black] {\rel{\audita{\act{read}}}} (f3);
                \draw[audit-a,thick,->] (u1) to node[swap,auto,near end,color=black] {\rel{\audita{\act{read}}}} (f4);
                \draw[interest-a,thick,->] (u1) to node[color=black,fill=white,inner sep=1pt] {\rel{\interesta}} (c3);
              \end{scope}
      \end{tikzpicture}
        \label{img:system-graph-fragment4:c}
    }
    \caption{Enforcing the Chinese Wall policy in RPPM}\label{img:system-graph-fragment4}
\end{figure}
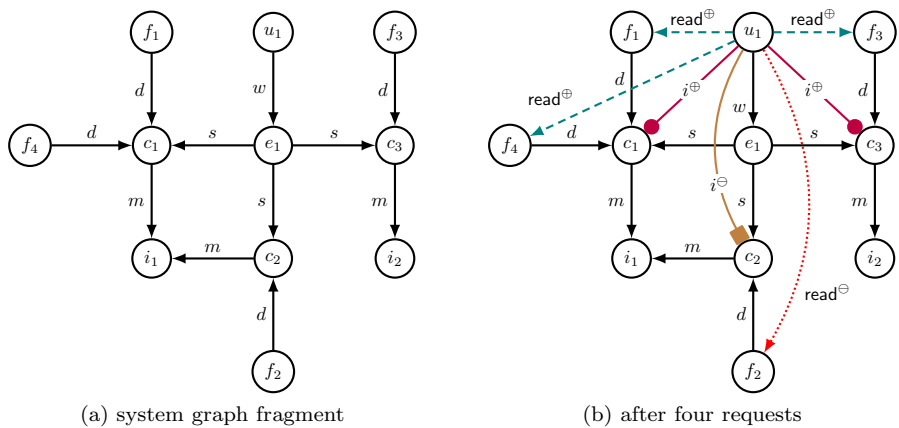

\section{Related Work}\label{sec:related-work}
Our work is motivated by the limitations of role-based access control in respect of the context in which authorization requests are made.
A medical doctor should not be given access to all patients' records simply because she is a doctor; it is only her patients' records to which she requires access.
We believe, as others do, that the relationships between entities within a system are intrinsic to the decision making process~\cite{CarminatiFP09,ChengPS12passat,ChengPS12dbsec,Fong11,ZhangAGC09}.
Whilst this approach has received considerable interest, much of the previous work has focused on applying relationship-based access control to online social networks.
This is no doubt due to the obvious alignment between the relationships in the access control model and the strong focus on interpersonal relationships in such networks.
As we have shown here, however, relationship-based access control can be applied to general-purpose computing systems (as well as in social networks).

Early work in this field focused solely on social networks and considered {\sf friend} and {\sf friend-of-a-friend} relationships to determine access to resources~\cite{KrukGGWC06}, as well as trust relationships between users~\cite{AlViMa07}.
These ideas were drawn together by Carminati~et~al.~\cite{CarminatiFP09} to create an access control model for social networks based on relationships.
This work has more recently been built upon by Hu~et~al.~~\cite{HuAhJo13} to provide joint management of access policies (once again focused on social networks).

Fong~et~al.~provided a richer policy language, based on modal logic~\cite{BrunsFSH12,Fong11} and then hybrid logic~\cite{Fong_ESORICS09}.
This policy language is not directly comparable to that of the RPPM model but there are common elements: both support defaults, relationship labels, positive and negative requirements, and can encode alternation (since the RPPM model may use multiple principal-matching rules for the same principal).
However, the RPPM model's policy language supports unbounded path conditions, which are useful when traversing a sub-graph comprising similar types of elements with paths that may be arbitrarily long (as in a directory tree, for example).
Moreover, we allow for arbitrary types of nodes in the system graph, which allows us to model relationships in general-purpose computing systems.

In terms of policy specification, the work of Cheng~et~al. is the most similar to our own~\cite{ChengPS12passat,ChengPS12dbsec}.
Whilst the focus of their work is again social networks, it allows for the specification of user-to-resource relationships other than ownership.
Our RPPM model is more general still in its support for entities of any kind (including logical ones) and policies not focused on, but still applicable to, social networks.
The policy language used by Cheng~et~al.~is based on a limited form of regular expressions, but limits the Kleene operators to a fixed, short, depth.
It is claimed that such restrictions are appropriate in social networks because of the ``six-degrees-of-separation'' phenomenon, which means the diameter of a social network is very small (compared to the number of nodes it contains).
However, no such assumptions can be made for more general systems, such as those to which the RPPM model can be applied.
We do not bound the Kleene plus operator and are able to support the Kleene star operator through the use of two principal-matching rules.
Finally, Cheng~et~al. do not allow cycles in the social network, when considering request evaluation with respect to their policies.
We see no reason why nodes should not be revisited and make no such restriction.

There has been some interest in recent years in reusing, recycling or caching authorization decisions at policy enforcement points in order to avoid recomputing decisions~\cite{BordersZP05,KohlerBS09,KohlerF09,WeiCB11}.
These techniques have perceived benefit, in particular, in large-scale, distributed, systems due to demands for reduced latency and a resilience to intermittent communications failures.
Whilst caching in the RPPM model does not resolve connectivity issues directly, the capability has considerable impact on latency.
Significantly, caching edges have direct value in the RPPM model, allowing the computationally expensive part of the decision-making process to be bypassed.
Further, a cached edge applies to multiple requests (as it considers the participants but not the action) and so can optimise processing at a level of abstraction above that which is commonly employed by other strategies.
This allows for more intuitive and less specific optimisation strategies and algorithms, which have demonstrable value (as shown in preliminary experimental work by Crampton and Sellwood~\cite{CramptonS14}).
Finally, whilst many authorization recycling strategies involve the policy enforcement point maintaining its own cache, or employing a ``speculative'' system~\cite{KiniB13}, caching in the RPPM model updates the system graph itself and can thus be implemented by very simple extensions to the basic model.

The RPPM model's support for auditing edges provides a natural mechanism through which to record past activity and thus inform future authorization requests.
Brewer and Nash's seminal paper on the Chinese Wall policy~\cite{BrewerN89} led to considerable research into history-based access control, which continues today.
Fong~et~al.~\cite{FoMeKr13} recently proposed a relationship-based model that incorporates temporal operators, enabling the specification and enforcement of history-based policies.
Once again, this model was developed in the context of social networks and cannot support the more general applications for which the RPPM model can be used (such as the Chinese Wall policy).

\section{Conclusion}\label{sec:conclusion}
We have introduced the RPPM model for access control based on the concepts of relationships, paths and principal matching.
We make use of relationships within our graph-based model to make authorization decisions.
By allowing the system graph to contain logical, as well as concrete entities, we enable contextual relationships to inform the decisions.
This is further demonstrated by the model's support for typed edges, thus allowing historical activities to influence the evaluation of current requests.
Through the use of these features, the RPPM model is able to easily accommodate our motivating examples from Section~\ref{sec:intro}.

Our model has rigorous foundations, using polices based on the concept of a path condition---which may be viewed as a restricted form of regular expression.
This, in turn, means we can describe the algorithm for evaluating access requests in terms of non-deterministic finite automata and determine its complexity.
Moreover, the expressiveness of path conditions does not need to be artificially constrained and we can thus represent a much wider range of policies than is possible with other proposals in the literature.
The flexibility of path conditions, comprising a mandated and precluded target, allows for the definition of practical and powerful authorization policies which are intuitive to interpret.
This allows for general computing applications such as file-systems and other tree structures (of arbitrarily large depth) to be modelled effectively.
Moreover, the model's support for caching edges enables the authorization system to bypass the expensive principal-matching algorithm in cases where it has already been performed.
The benefit of caching in the RPPM model is tangible and significant whilst incurring little overhead to implement.

The generality and flexibility of the RPPM model makes it an ideal basis upon which further work can be built.
Audit edges, as well as providing support for the enforcement of separation-of-duty and Chinese wall policies, introduce a natural route into workflow task authorization and stateful resources to which access changes over time.
We also plan to develop an administrative model for the RPPM model by including administrative principals and rules for matching administrative principals to requests that modify the state of the system (such as adding nodes and edges to the system graph).
More speculatively, we hope to consider strategies for partitioning the system graph into sub-graphs each having their own policies and using ``hub'' entities to create ``bridges'' between the sub-graphs.
The intuition is that segmentation of the system graph may hold the key to more space-efficient policy representation, storage and retrieval, as well as more time-efficient request evaluation.

\bibliographystyle{abbrv}
\bibliography{../RPPM}

\begin{thebibliography}{10}

\bibitem{AbadiF03}
M.~Abadi and C.~Fournet.
\newblock Access control based on execution history.
\newblock In {\em NDSS}. The Internet Society, 2003.

\bibitem{Aho_Compilers}
A.~V. Aho, R.~Sethi, and J.~D. Ullman.
\newblock {\em Compilers: Principles, Techniques, and Tools}.
\newblock Addison-Wesley Longman Publishing Co., Inc., Boston, MA, USA, 1986.

\bibitem{Al-KahtaniS02}
M.~A. Al{-}Kahtani and R.~S. Sandhu.
\newblock A model for attribute-based user-role assignment.
\newblock In {\em 18th Annual Computer Security Applications Conference
  {(ACSAC} 2002), 9-13 December 2002, Las Vegas, NV, {USA}}, pages 353--362.
  {IEEE} Computer Society, 2002.

\bibitem{AlViMa07}
B.~Ali, W.~Villegas, and M.~Maheswaran.
\newblock A trust based approach for protecting user data in social networks.
\newblock In K.~A. Lyons and C.~Couturier, editors, {\em CASCON}, pages
  288--293. IBM, 2007.

\bibitem{BeckerFG10}
M.~Y. Becker, C.~Fournet, and A.~D. Gordon.
\newblock Secpal: Design and semantics of a decentralized authorization
  language.
\newblock {\em Journal of Computer Security}, 18(4):619--665, 2010.

\bibitem{BertinoBF01}
E.~Bertino, P.~A. Bonatti, and E.~Ferrari.
\newblock {TRBAC:} {A} temporal role-based access control model.
\newblock {\em {ACM} Trans. Inf. Syst. Secur.}, 4(3):191--233, 2001.

\bibitem{BordersZP05}
K.~Borders, X.~Zhao, and A.~Prakash.
\newblock Cpol: high-performance policy evaluation.
\newblock In V.~Atluri, C.~Meadows, and A.~Juels, editors, {\em ACM Conference
  on Computer and Communications Security}, pages 147--157. ACM, 2005.

\bibitem{BrewerN89}
D.~F.~C. Brewer and M.~J. Nash.
\newblock The {Chinese Wall} security policy.
\newblock In {\em IEEE Symposium on Security and Privacy}, pages 206--214. IEEE
  Computer Society, 1989.

\bibitem{BrunsFSH12}
G.~Bruns, P.~W.~L. Fong, I.~Siahaan, and M.~Huth.
\newblock Relationship-based access control: its expression and enforcement
  through hybrid logic.
\newblock In E.~Bertino and R.~S. Sandhu, editors, {\em Second {ACM} Conference
  on Data and Application Security and Privacy, {CODASPY} 2012, San Antonio,
  TX, USA, February 7-9, 2012}, pages 117--124. {ACM}, 2012.

\bibitem{BrunsH11}
G.~Bruns and M.~Huth.
\newblock Access control via {Belnap} logic: Intuitive, expressive, and
  analyzable policy composition.
\newblock {\em {ACM} Trans. Inf. Syst. Secur.}, 14(1):9, 2011.

\bibitem{CarminatiFP09}
B.~Carminati, E.~Ferrari, and A.~Perego.
\newblock Enforcing access control in web-based social networks.
\newblock {\em ACM Trans. Inf. Syst. Secur.}, 13(1), 2009.

\bibitem{ChengPS12passat}
Y.~Cheng, J.~Park, and R.~S. Sandhu.
\newblock Relationship-based access control for online social networks: Beyond
  user-to-user relationships.
\newblock In {\em SocialCom/PASSAT}, pages 646--655. IEEE, 2012.

\bibitem{ChengPS12dbsec}
Y.~Cheng, J.~Park, and R.~S. Sandhu.
\newblock A user-to-user relationship-based access control model for online
  social networks.
\newblock In N.~Cuppens-Boulahia, F.~Cuppens, and J.~Garc\'{\i}a-Alfaro,
  editors, {\em DBSec}, volume 7371 of {\em Lecture Notes in Computer Science},
  pages 8--24. Springer, 2012.

\bibitem{CramptonM12}
J.~Crampton and C.~Morisset.
\newblock {PTaCL}: {A} language for attribute-based access control in open
  systems.
\newblock In P.~Degano and J.~D. Guttman, editors, {\em Principles of Security
  and Trust - First International Conference, {POST} 2012, Held as Part of the
  European Joint Conferences on Theory and Practice of Software, {ETAPS} 2012,
  Tallinn, Estonia, March 24 - April 1, 2012, Proceedings}, volume 7215 of {\em
  Lecture Notes in Computer Science}, pages 390--409. Springer, 2012.

\bibitem{CramptonS14_STM}
J.~Crampton and J.~Sellwood.
\newblock Caching and auditing in the {RPPM} model.
\newblock In S.~Mauw and C.~D. Jensen, editors, {\em Security and Trust
  Management - 10th International Workshop, {STM} 2014, Wroclaw, Poland,
  September 10-11, 2014. Proceedings}, volume 8743 of {\em Lecture Notes in
  Computer Science}, pages 49--64. Springer, 2014.

\bibitem{CramptonS14}
J.~Crampton and J.~Sellwood.
\newblock Path conditions and principal matching: a new approach to access
  control.
\newblock In S.~L. Osborn, M.~V. Tripunitara, and I.~Molloy, editors, {\em
  SACMAT}, pages 187--198. ACM, 2014.

\bibitem{DamianiBCP07}
M.~L. Damiani, E.~Bertino, B.~Catania, and P.~Perlasca.
\newblock {GEO-RBAC:} {A} spatially aware {RBAC}.
\newblock {\em {ACM} Trans. Inf. Syst. Secur.}, 10(1), 2007.

\bibitem{EdjlaliAC99}
G.~Edjlali, A.~Acharya, and V.~Chaudhary.
\newblock History-based access control for mobile code.
\newblock In J.~Vitek and C.~D. Jensen, editors, {\em Secure Internet
  Programming}, volume 1603 of {\em Lecture Notes in Computer Science}, pages
  413--431. Springer, 1999.

\bibitem{Fong11}
P.~W.~L. Fong.
\newblock Relationship-based access control: protection model and policy
  language.
\newblock In R.~S. Sandhu and E.~Bertino, editors, {\em CODASPY}, pages
  191--202. ACM, 2011.

\bibitem{Fong_ESORICS09}
P.~W.~L. Fong, M.~M. Anwar, and Z.~Zhao.
\newblock A privacy preservation model for facebook-style social network
  systems.
\newblock In M.~Backes and P.~Ning, editors, {\em ESORICS}, volume 5789 of {\em
  Lecture Notes in Computer Science}, pages 303--320. Springer, 2009.

\bibitem{FoMeKr13}
P.~W.~L. Fong, P.~Mehregan, and R.~Krishnan.
\newblock Relational abstraction in community-based secure collaboration.
\newblock In A.-R. Sadeghi, V.~D. Gligor, and M.~Yung, editors, {\em ACM
  Conference on Computer and Communications Security}, pages 585--598. ACM,
  2013.

\bibitem{GiuriI97}
L.~Giuri and P.~Iglio.
\newblock Role templates for content-based access control.
\newblock In {\em {ACM} Workshop on Role-Based Access Control}, pages 153--159,
  1997.

\bibitem{GligorGF98}
V.~D. Gligor, S.~I. Gavrila, and D.~F. Ferraiolo.
\newblock On the formal definition of separation-of-duty policies and their
  composition.
\newblock In {\em IEEE Symposium on Security and Privacy}, pages 172--183. IEEE
  Computer Society, 1998.

\bibitem{GurevishN08}
Y.~Gurevich and I.~Neeman.
\newblock {DKAL:} distributed-knowledge authorization language.
\newblock In {\em Proceedings of the 21st {IEEE} Computer Security Foundations
  Symposium, {CSF} 2008, Pittsburgh, Pennsylvania, 23-25 June 2008}, pages
  149--162. {IEEE} Computer Society, 2008.

\bibitem{HuAhJo13}
H.~Hu, G.-J. Ahn, and J.~Jorgensen.
\newblock Multiparty access control for online social networks: Model and
  mechanisms.
\newblock {\em IEEE Trans. Knowl. Data Eng.}, 25(7):1614--1627, 2013.

\bibitem{KhanF12}
A.~A. Khan and P.~W.~L. Fong.
\newblock Satisfiability and feasibility in a relationship-based workflow
  authorization model.
\newblock In S.~Foresti, M.~Yung, and F.~Martinelli, editors, {\em Computer
  Security - {ESORICS} 2012 - 17th European Symposium on Research in Computer
  Security, Pisa, Italy, September 10-12, 2012. Proceedings}, volume 7459 of
  {\em Lecture Notes in Computer Science}, pages 109--126. Springer, 2012.

\bibitem{KiniB13}
P.~Kini and K.~Beznosov.
\newblock Speculative authorization.
\newblock {\em {IEEE} Trans. Parallel Distrib. Syst.}, 24(4):814--824, 2013.

\bibitem{KohlerBS09}
M.~Kohler, A.~D. Brucker, and A.~Schaad.
\newblock Proactive caching: Generating caching heuristics for business process
  environments.
\newblock In {\em CSE (3)}, pages 297--304. IEEE Computer Society, 2009.

\bibitem{KohlerF09}
M.~Kohler and R.~Fies.
\newblock Proactive caching - a framework for performance optimized access
  control evaluations.
\newblock In {\em POLICY}, pages 92--94. IEEE Computer Society, 2009.

\bibitem{KrukGGWC06}
S.~R. Kruk, S.~Grzonkowski, A.~Gzella, T.~Woroniecki, and H.~Choi.
\newblock {D-FOAF:} distributed identity management with access rights
  delegation.
\newblock In R.~Mizoguchi, Z.~Shi, and F.~Giunchiglia, editors, {\em The
  Semantic Web - {ASWC} 2006, First Asian Semantic Web Conference, Beijing,
  China, September 3-7, 2006, Proceedings}, volume 4185 of {\em Lecture Notes
  in Computer Science}, pages 140--154. Springer, 2006.

\bibitem{KrukowNS08}
K.~Krukow, M.~Nielsen, and V.~Sassone.
\newblock A logical framework for history-based access control and reputation
  systems.
\newblock {\em Journal of Computer Security}, 16(1):63--101, 2008.

\bibitem{Lampson_Protection}
B.~W. Lampson.
\newblock Protection.
\newblock In {\em Proceedings of the Fifth Princeton Conference on Information
  Sciences and Systems}. Princeton, 1971.

\bibitem{XACML3}
OASIS.
\newblock {\em eXtensible Access Control Markup Language (XACML) Version 3.0},
  2010.
\newblock OASIS Committee Specification 01 (Erik Rissanen, editor).

\bibitem{SandhuP03}
R.~S. Sandhu and J.~Park.
\newblock Usage control: A vision for next generation access control.
\newblock In V.~Gorodetsky, L.~J. Popyack, and V.~A. Skormin, editors, {\em
  MMM-ACNS}, volume 2776 of {\em Lecture Notes in Computer Science}, pages
  17--31. Springer, 2003.

\bibitem{SimonZ97}
R.~T. Simon and M.~E. Zurko.
\newblock Separation of duty in role-based environments.
\newblock In {\em CSFW}, pages 183--194. IEEE Computer Society, 1997.

\bibitem{WeiCB11}
Q.~Wei, J.~Crampton, K.~Beznosov, and M.~Ripeanu.
\newblock Authorization recycling in hierarchical rbac systems.
\newblock {\em ACM Trans. Inf. Syst. Secur.}, 14(1):3, 2011.

\bibitem{ZhangAGC09}
R.~Zhang, A.~Artale, F.~Giunchiglia, and B.~Crispo.
\newblock Using description logics in relation based access control.
\newblock In B.~C. Grau, I.~Horrocks, B.~Motik, and U.~Sattler, editors, {\em
  Description Logics}, volume 477 of {\em CEUR Workshop Proceedings}.
  CEUR-WS.org, 2009.

\end{thebibliography}

\end{document}